\documentclass{pnastwo}
\usepackage{graphicx}
\usepackage{amssymb,amsfonts,amsmath}
\usepackage{url}
\usepackage{framed}

\usepackage{epstopdf} 
\begin{document}

\title{Quantifying the benefits of vehicle pooling with shareability networks}

\author{
Paolo Santi\affil{1}{Senseable City Laboratory, Massachusetts Institute of Technology, 77 Massachusetts Avenue, Cambridge, MA 02139}\affil{2}{Istituto di Informatica e Telematica del CNR, Pisa 56124, Italy}, Giovanni Resta\affil{2}{Istituto di Informatica e Telematica del CNR, Pisa 56124, Italy}, Michael Szell\affil{1}{Senseable City Laboratory, Massachusetts Institute of Technology, 77 Massachusetts Avenue, Cambridge, MA 02139}, Stanislav Sobolevsky\affil{1}{Senseable City Laboratory, Massachusetts Institute of Technology, 77 Massachusetts Avenue, Cambridge, MA 02139}, Steven Strogatz\affil{3}{Dept. of Mathematics, Cornell University, Ithaca, NY 14853}, Carlo Ratti\affil{1}{Senseable City Laboratory, Massachusetts Institute of Technology, 77 Massachusetts Avenue, Cambridge, MA 02139}
}

\contributor{}

\maketitle

\begin{article}

\begin{abstract}
Taxi services are a vital part of urban transportation, and a considerable contributor to traffic congestion and air pollution causing substantial adverse effects on human health. Sharing taxi trips is a possible way of reducing the negative impact of taxi services on cities, but this comes at the expense of passenger discomfort quantifiable in terms of a longer travel time. Due to computational challenges, taxi sharing has traditionally been approached on small scales, such as within airport perimeters, or with dynamical ad-hoc heuristics. However, a mathematical framework for the systematic understanding of the tradeoff between collective benefits of sharing and individual passenger discomfort is lacking. Here we introduce the notion of shareability network which allows us to model the collective benefits of sharing as a function of passenger inconvenience, and to efficiently compute optimal sharing strategies on massive datasets. We apply this framework to a dataset of millions of taxi trips taken in New York City, showing that with increasing but still relatively low passenger discomfort, cumulative trip length can be cut by 40\% or more. This benefit comes with reductions in service cost, emissions, and with split fares, hinting towards a wide passenger acceptance of such a shared service. Simulation of a realistic online system demonstrates the feasibility of a shareable taxi service in New York City. Shareability as a function of trip density saturates fast, suggesting effectiveness of the taxi sharing system also in cities with much sparser taxi fleets or when willingness to share is low.\end{abstract}


\begin{framed}
\subsection{Significance}
{\bf \sf Recent advances in information technologies have increased our participation in ``sharing economies'', where applications that allow networked, real-time data exchange facilitate the sharing of living spaces, equipment, or vehicles with others. However, the impact of large-scale sharing on sustainability is not clear, and a framework to assess quantitatively its benefits is missing. For this purpose, we propose the method of shareability networks which translates spatio-temporal sharing problems into a graph-theoretic framework that provides efficient solutions. Applying this method to a dataset of 150 million taxi trips in New York City, our simulations reveal the vast potential of a new taxi system in which trips are routinely shareable while keeping passenger discomfort low in terms of prolonged travel time.}
\end{framed}

\dropcap{V}ehicular traffic congestion -- and the air pollution that results from it -- is one of the greatest challenges facing cities all over the world. It comes at great monetary and human cost: in the 83 largest urban areas of the US alone, the amount of wasted time and fuel caused by congestion has been placed at US\$ 60 billion \cite{arnott1994etc}. At the same time, the World Health Organization has estimated that over one million deaths per year worldwide can be attributed to outdoor air pollution \cite{world2011whs}, which is to a large part caused by vehicular traffic \cite{caiazzo2013ape}. Further adverse effects include fatalities through road accidents and economic losses from missed business activities. For these reasons, great hope is placed today in the rapid deployment of digital information and communication technologies that could help make cities ``smarter'' \cite{batty2013nsc}, and, in particular, that could help manage vehicular traffic more efficiently. The use of real-time information allows the monitoring of the urban mobility infrastructure to an unprecedented extent, and opens up new potential for the exploitation of unused capacity. One major example is the public mobility infrastructure: taking advantage of the wide-spread use of smart phones and their capabilities for running real-time applications (apps) it is possible to design new, smarter transportation systems based on the sharing of cars or minivans, effectively providing services that could replace public transportation with the on-demand qualities of individual mobility or taxis \cite{dowling2013swy}. However, while this option has been proposed in the past, municipal authorities, city residents, and other stakeholders may be reluctant to invest in it until its benefits have been quantified \cite{handke2013fr}. This is the goal of the present paper.

\begin{figure*} [b]
\centering
\includegraphics[width=12cm]{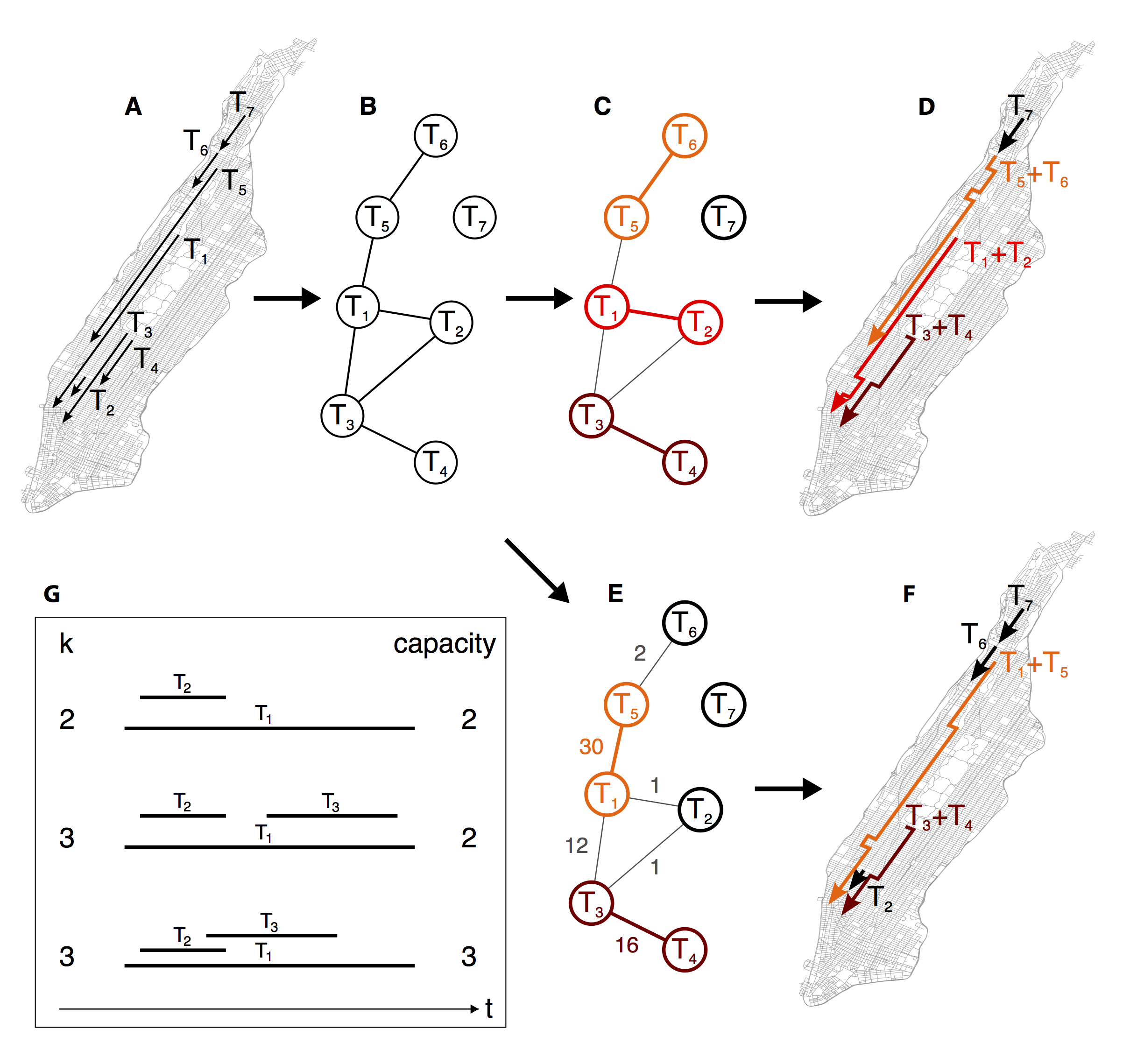}
\caption{ Shareability networks translate spatio-temporal sharing problems into a graph-theoretic framework that provides efficient solutions. (A), Example of seven trips, $T_1, \ldots, T_7$, requested and to be shared in Manhattan, New York City. (B), Construction of shareability network for $k=2$. Trips that could potentially be shared are connected, given the necessary time constraints to hold which we assume here to be the case. Trips 1 and 4 cannot be shared because the total length of the best shared route would be longer than the sum of the single routes. Likewise, trip 7 is an isolated node because it cannot possibly be shared with other trips. (C), Maximum matching of the shareability network gives the maximum number of trip pairs, i.e. the maximum number of shared trips. (D), Implementation (routing) of the maximum matching solution. (E), Alternatively, maximum weighted matching of the shareability network gives the solution with the minimal total travel time, which in this case leads to a different solution than unweighted maximum matching. Here, only two pairs of trips are shared, but the amount of travel time saved, given by the sum of link weights of the matching, $30+16$, is optimal. (F), Implementation (routing) of the weighted maximum matching solution. (G), $k$-sharing and taxi capacity. Each of the three cases involves a number of trips $T_i$ to be shared, but ordered differently in time $t$. The top case corresponds to a feasible sharing according to our model with $k=2$, and the trips can be accommodated in a taxi with capacity $\ge 2$. The middle case corresponds to a model with $k=3$ since three trips are combined, but the three trips can be combined in a taxi with capacity two since two of the trips are non-overlapping. The bottom case corresponds to $k=3$, but here a taxi capacity $\ge 3$ is needed to accommodate the combined trips.}
\end{figure*}

At the basis of a shared taxi service is the concept of ride-sharing or carpooling, a long-standing proposition for decreasing road traffic, which originated during the ``oil crisis" in the 1970s \cite{handke2013fr}. During that time, economic incentives outbalanced the psychological barriers on which successful carpooling programs depend: giving up personalized transportation and accepting strangers in the same vehicle. Surveys indicate that the two most important deterrents to potential carpoolers are the extra time requirements and the loss of privacy \cite{dueker1977rsp, teal1987cwh}. However, the lack of correlations between socio-demographic variables and carpooling propensity \cite{teal1987cwh}, the design of appropriate economic incentives \cite{benakiva1977mst}, and recent practical implementations of taxi sharing systems in New York City (http://bandwagon.io/) give ample hope that many social obstacles might be overcome in newly emerging ``sharing economies'' \cite{botsman2010wmy, john2013social}.

Besides psychological considerations, it is fundamental to understand the logistic limitations of realistic taxi sharing systems, which is our focus here. From a theoretical perspective, trip-sharing is traditionally seen as an instance of ``dynamic pickup and delivery'' problems \cite{yang2004rmt,Berbeglia10}, in which a number of goods or customers must be picked up and delivered efficiently at specific locations within well-defined time windows. Such problems are typically solved by means of linear programming, in which a function of the system variables is optimized subject to a set of equations that describe the constraints. While linear programming tasks can be solved with standard approaches of Operations Research or with constraint programming \cite{shaw1998ucp}, their computational feasibility heavily depends on the number of variables and equations, e.g., the pickup and delivery time windows of each customer, used to describe the problem at hand. Most previous taxi studies have therefore focused on small-scale routing problems, such as within airport perimeters \cite{marin2006amt,keith2008otr}. Large urban taxi systems, in contrast, involve thousands of vehicles performing hundreds of thousands of trips per day. A first step towards practical taxi ride sharing systems is \cite{TShare}, where the authors present the design of a dynamic ride sharing system inclusive of a taxi dispatching strategy and fare management. Due to computational reasons trip sharing in \cite{TShare} is decided based on a heuristic approach tailored to the specific taxi dispatching strategy at hand. Our approach, by contrast, is the development of a framework which enables to investigate \emph{in general terms} the fundamental trade-off between the benefit and the passenger discomfort induced by taxi-sharing systems at the city level, as an example from a wide class of spatial sharing problems.

Here we introduce the notion of shareability network to model trip sharing in a simple static way, and apply classical methods from graph theory to solve the taxi trip-sharing problem in a provably efficient way. The differences between static trip sharing as considered herein, and dynamic sharing as considered, e.g., in \cite{TShare}, are discussed in detail in the SI Appendix. The starting point of our analysis is a dataset composed of the records of over 150 million taxi trips originating and ending in Manhattan in the year 2011 by all 13,586 registered taxis. For each trip, the record reports the vehicle ID, the Global Positioning System (GPS) coordinates of the pickup and drop-off locations, and corresponding times. Pickup and drop-off locations have been associated with the closest street intersection in the road map of Manhattan (Materials and Methods). We impose a natural network structure on an otherwise unstructured, gigantic search space of the type explored in traditional linear programming. To this end we define two parameters: the {\em shareability parameter} $k$, standing for the maximum number of trips that can be shared, and the {\em quality of service parameter} $\Delta$, which stands for the maximum delay a customer tolerates in a shared taxi service trip, mathematically equivalent to the notion of ``time window" used in other approaches \cite{Berbeglia10,TShare}. To ease the analysis, we use the $\Delta$ formalism, however, when presented in a real implementation to passengers, it might be psychologically more effective to use the neutral wording ``time window'' rather than explicitly mentioning the maybe more negatively connoted word ``delay''. The choice of defining the quality of service parameter as an absolute time, instead of as a percentage increase of the travel time, is in line with similar realizations in the literature \cite{TShare}, and is motivated by the fact that absolute delay information is likely more valuable than percent estimation of travel time increase for potential customers of a shared taxi service. Further, let $T_i=(o_i,d_i,t^o_i,t^d_i), i=1\ldots k$ be $k$ trips where $o_i$ denotes the origin of the trip, $d_i$ the destination, and $t_i^o,t_i^d$ the starting and ending times, respectively. We say that multiple trips $T_i$ are {\em shareable} if there exists a route connecting all the $o_i$ and $d_i$ in any order where each $o_i$ precedes the corresponding $d_i$, except for configurations where single trips are concatenated and not overlapped like $o_1 \rightarrow d_1 \rightarrow  o_2 \rightarrow  d_2$, such that each customer is picked up and dropped at the respective origin and destination locations with delay at most $\Delta$, with the delay computed as the time difference to the respective single, individual trip. Imposing a bound of $k$ on shareability implies that the $k$ trips can be combined using a taxi of corresponding capacity (Fig.~1G). Deciding whether two or more trips can be shared necessitates knowledge of the travel time between arbitrary intersections in Manhattan, which we estimated using an ad-hoc heuristic (Fig.~S2, Table S1).

For the case $k=2$, the {\em shareability network} associated with a set $\mathcal{T}$ of trips is obtained by assigning a node $T$ for each trip in $\mathcal{T}$, and by placing a link between two nodes $T_i$ and $T_j$ if the two trips can be shared for the given value of $\Delta$ (Fig.~1A and B). The value of $\Delta$ has a profound impact on topological properties of the resulting shareability network. Increasing $\Delta$ capitalizes on well-known effects of time-aggregated networks such as densification \cite{kossinets2006eae,leskovec2007ged}, capturing the intuitive notion that, the more patient the customers, the more opportunities for trip sharing arise (Fig.~2A and B). For values of $k>2$, the shareability network has a hyper-graph structure in which up to $k$ nodes can be connected by a link simultaneously. Because of computational reasons, the shareability parameter $k$ has a substantial impact on the feasibility of solving the problem. A solution is tractable for $k=2$, heuristically feasible for $k=3$, while it becomes computationally intractable for $k\geq4$ (SI Appendix). This constraint implies that taxi sharing services, and social sharing applications in general, will likely be able to combine only a limited number of trips. However, as we show below, even the minimum possible number of trip combinations ($k=2$) can provide immense benefits to a dense enough community like the City of New York.

 \begin{figure}[b]
\centering
\includegraphics[width=0.49\textwidth]{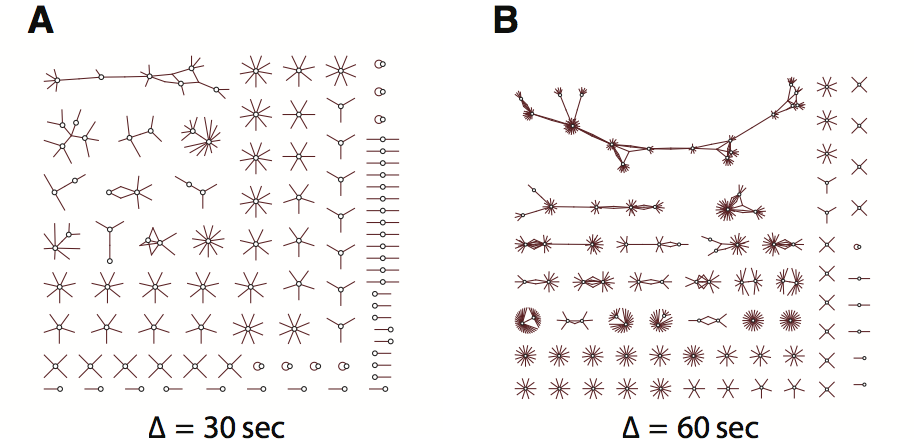}
\caption{Shareability networks densify with longer time aggregation, increasing sharing opportunities. This exemplary subset of the shareability network corresponds to 100 consecutive trips for values of (A), $\Delta=30\,\mathrm{sec}$ and (B), $\Delta=60\,\mathrm{sec}$. Open links point to trips outside the considered set of trips. Isolated nodes are represented as self-loops. Node positions are not preserved across the networks. A similar, although visually not insightful, densification
effect is observed in shareability networks obtained when $k = 3$.}
 \end{figure}
 
With the shareability network, classical algorithms for solving maximum matching on graphs \cite{Cormen90,Galil86} can be used to determine the best trip sharing strategy according to two optimization criteria: $a)$ maximizing the number of shared trips, or $b)$ minimizing the cumulative time needed to accommodate all trips. To find the best solution according to $a)$ or $b)$, it is sufficient to compute a maximum matching or a weighted maximum matching on the shareability network, respectively (Fig.~1C and E, Materials and Methods). Since a shared trip can be served by a single taxi instead of two, the number of shared trips can be used as a proxy for the reduction in number of circulating taxis. For instance, an $80\%$ rate of shared trips translates into a $40\%$ reduction of the taxi fleet. Other important objectives such as total system cost and emissions are reasonably approximated by criterion $b)$.

\section{Results}

Using a maximum value of $\Delta=10\,\mathrm{min}$ and all trips performed in New York City in the year 2011, the resulting shareability network has more than 150 million nodes and over 100 billion links. We first consider trip sharing opportunities under a model in which the entire shareability network is known beforehand, and maximum matchings are computed on the entire network. This omniscient \emph{Oracle} approach models an artificial scenario in which trip sharing decisions can be taken considering not only the current taxi requests, but also all future ones, serving as a theoretical upper bound for sharing opportunities. In practice, the Oracle model is useful to assess the benefits of social sharing systems where bookings are placed well ahead of time (Fig.~3A). Because of this foreknowledge, even with the low and reasonable value of $\Delta=2\,\mathrm{min}$, the average percentage of shareable trips is close to $100\%$ (Fig.~3B).

In practical systems however, the Oracle approach is of limited use, as only trip requests issued in a relatively short time window are known at decision time, corresponding to a small time-slice of the shareability network. In the following, we therefore focus on trip sharing opportunities in a realistic model in which the trip sharing decision for a trip $T_1$ considers only trips which start within a short interval around its starting time $t^o_1$. More formally, we retain in the shareability network only links connecting trips $T_i$ and $T_j$ such that $|t^o_i-t^o_j|\le \delta$, where $\delta$ is a {\em time window} parameter. This \emph{Online} model is representative of a scenario in which a customer using an ``e-hailing" application issues a taxi request reporting pickup and drop-off locations, and after the small time window $\delta$ receives feedback from the taxi management system whether a shared ride is available. This parameter is fundamental in the Online model: the larger $\delta$, the more trip sharing opportunities can be exploited, for the same reasons of network time-aggregation as with $\Delta$ (Fig.~S3). However, $\delta$ should be kept reasonably small to be acceptable by a potential customer, and to allow real-time computation of the shared trip matching (SI Appendix). Therefore, in what follows, we set $\delta = 1\,\mathrm{min}$. 

\begin{figure*}[t]
\centering
\includegraphics[width=14cm]{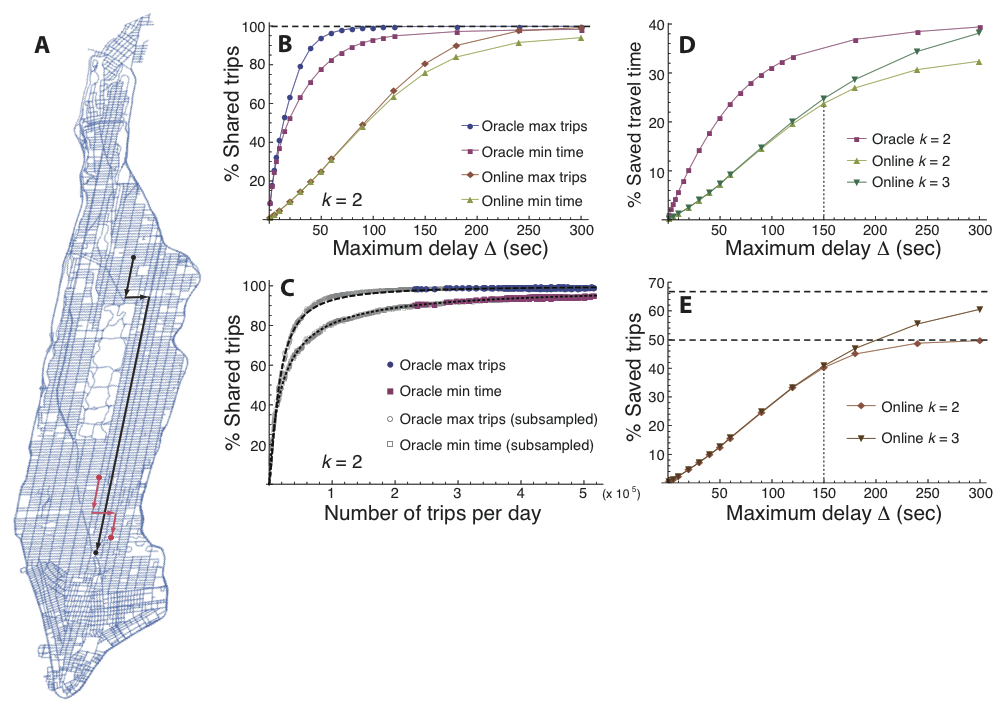}
\caption{ Benefits of trip sharing. (A), Street network of Manhattan, and examples of two trips that can be shared under the omniscient Oracle model, but not under the Online model. The starting time of the red trip is much later than that of the black trip, but in the Online model trip sharing decisions must be taken within a very short time window $\delta = 1\,\mathrm{min}$ to notify customers of trip sharing opportunities as soon as possible after their order. (B), Percentage of shared trips as a function of the trip time delay $\Delta$ in the Oracle and in the Online model for the two considered optimization criteria of maximizing shared trips (max trips) and minimizing total travel time (min time), when up to $k=2$ trips can be shared. (C), Shareability as a function of trips per day in the Oracle model. Typical days in New York City feature around 400,000 trips with near maximum shareability. Subsampling data by randomly removing vehicles reveals the underlying saturation curves, fit (dashed lines) by a simple function of type $f(x) = \frac{Kx^n}{1+Kx^n}$ with the two parameters $K$ and $n$ well-known in adsorption processes and biochemical systems (SI Appendix). The fast, hyperbolic saturation implies that taxi sharing could be effective even in cities with vehicle densities much lower than New York, or when the willingness to share is low. (D), Percentage of saved travel time as a function of $\Delta$ for $k=2$ and $k=3$. Although $\delta$ is reduced from practically infinite in the Oracle model to $\delta=1\,\mathrm{min}$ in the Online model, saved travel time is well above $30\%$ for $\Delta = 300\,\mathrm{sec}$, for $k=3$ almost reaching the maximum possible value from the Oracle model with $k=2$. (E), Percentage of saved trips as a function of $\Delta$ with $k=2$ and $k=3$. The theoretically possible maximum (dashed lines) of $50\%$ for $k=2$ and $66.7\%$ for $k=3$ are closely approximated. For $\Delta < 150\,\mathrm{sec}$ (dotted line), the benefits of 3-sharing over 2-sharing are negligible.}
\end{figure*}

As expected, reducing the time horizon $\delta$ from practically infinite in the Oracle model to $1\,\mathrm{min}$ in the Online model considerably reduces trip sharing opportunities for low values of $\Delta$. For instance, when $\Delta=1\,\mathrm{min}$, the Oracle model allows sharing of $94.5\%$ of the trips, but the Online model only less than $30\%$. However, the situation is much less penalizing for the Online model when the delay parameter is increased within reasonable range. When $\Delta=5\,\mathrm{min}$, the Online model can exploit virtually {\em all} available trip sharing opportunities (Fig.~3B). Concerning saved travel time, results are similarly promising (Fig.~3D). When $\Delta=5\,\mathrm{min}$, we can save $32\%$ of total travel time with the Online model, compared to $40\%$ savings in the optimal Oracle model. Note that our method only concerns the sharing of non-vacant trips, but these make up the majority of taxi traffic \cite{Phitha10,FactBook}. In fact, the fraction of time during which taxis are serving customers corresponds to the high value of about 75\% of the on-service time of a taxi (SI Appendix and Fig.~S1). Accounting for the effect of empty trips thus would approximately reduce the total travel time savings from 40\% and 32\% to the still substantial values of 30\% and 24\% in the Oracle and Online model, respectively.

Is it possible to even further improve efficiency by increasing the number $k$ of shareable trips? When $k=3$, the shareability network becomes a shareability hyper-network, for which maximum matching is solvable only in approximation using a heuristic algorithm which is computationally feasible for relatively small networks only \cite{Chandra01,Johnson74}. Because of this methodological issue and the combinatorial explosion of sharing options, we calculated the number of shared trips and the fraction of saved travel time for $k=3$ only in the Online model -- which by definition features much smaller shareability networks. Simulations show that increasing the number of shareable trips $k$ provides noticeable benefits only when the quality of service parameter $\Delta$ crosses a threshold around $\Delta_{\mathrm{crit}} \approx 150\,\mathrm{sec}$ (Fig.~3D and E). When $\Delta=300\,\mathrm{sec}$, the number of saved taxi trips is increased from about $50\%$ with $k=2$ to about $60\%$ with $k=3$, which is however well below the $66.7\%$ maximum theoretical percentage of shared trips. This suboptimal result suggests that the effort for implementing a service for sharing $k>2$ trips may not be well justified. Further, to become widely accepted, a multi-shared taxi service might require vehicles of higher capacity and/or physically separated, private compartments, possibly inflating overhead for $k>2$. Since the benefit of multi-sharing is not that high, it might not cover these additional expenses.

\section{Discussion}

Our analysis shows that New York City offers ample opportunities for trip sharing with minimal passenger discomfort, without having to resort to a computationally demanding sharing strategy in which already started trips would be re-routed on the fly, and that these opportunities are realistic to be implemented in a new taxi system. From a computational standpoint, the polynomial runtimes of our algorithms suggest that there should be no issues with designing systems in which taxi companies calculate sharing options within $\delta=1\,\mathrm{min}$ of the request and immediately dispatch their taxis. By implementing a system that is $40\%$ more efficient and affordable, the ultimate goal is to make taxi systems a more attractive and sustainable mode of transportation, able to generate increased demand and to satisfy it with the current or an even higher number of vehicles.

In order to assess to which extent our results could be generalized to cities with lower taxi densities than New York, or to account for situations where willingness to share or where market penetration of an accompanying software app is low, we studied how the number of shareable trips in a given day changes as a function of the total number of trips (Fig. 3C). The average number of daily trips in New York is highly concentrated around 400,000. Hence, we have generated additional low density situations by subsampling the dataset, randomly removing increasing fractions of vehicles from the system (Materials and Methods). The resulting shareability values are excellently fit by saturation curves of the form $f(x) = \frac{Kx^n}{1+Kx^n}$. These curves are well-known to describe binding processes in biochemical systems, providing an interesting link to general pairing problems (SI Appendix). At around 100,000 trips, or 25\% of the daily average, we already reach saturation and near maximum shareability. This fast saturation suggests that taxi sharing systems could be effective even in cities with taxi fleet densities much lower than New York. 

Future work should aim to assess in more detail the psychological limitations of taxi sharing, to understand the conditions and appropriate incentive systems under which individuals are willing to be seated in the same vehicle. This includes the design of suitable faring systems aimed at fairly distributing the economic benefits of sharing between drivers and customers, such as the one proposed in \cite{TShare}. Moreover, the sharing analysis should be extended to other cities to better understand the generalizability of the results, and if possible, to measure and incorporate currently unknown data such as the actual search or waiting times of passengers who are trying to find an empty taxi, or the number of passengers that are being transported per vehicle. Finally, the framework of shareability networks could be used to study more generally other social sharing scenarios \cite{carvalho2012fsr} such as ride sharing of cars, bikes, etc.~or the communal usage of equipment which is characterized by considerable unit cost and infrequent use, stimulating new forms of sharing and models of ownership \cite{botsman2010wmy}.

\section{Materials and Methods}
{\sf \small
\subsection{Trip data}
The dataset contains origin-destination data of all 172 million trips with passengers of all 13,586 taxicabs in New York during the calendar year of 2011. There are 39,437 unique driver IDs in the dataset, which corresponds to 2.9 drivers per taxi on average. The dataset contains a number of fields from which we use the following: medallion ID, origin time, destination time, origin longitude, origin latitude, destination longitude, destination latitude. Times are accurate to the second, positional information has been collected via Global Positioning System (GPS) technology by the data provider. Out of our control are possible biases due to urban canyons which might have slightly distorted the GPS locations during the collection process \cite{grush2008cam}. All IDs are given in anonymized form, origin and destination values refer to the origins and destinations of trips, respectively. 

\subsection{Map data and map matching}
For creating the street network of Manhattan we used data from openstreetmap.org. We filtered the streets of Manhattan, selecting only the following road classes: primary, secondary, tertiary, residential, unclassified, road, living street. Several other classes were deliberately left out, such as footpaths, trunks, links or service roads, as they are unlikely to contain delivery or pickup locations. Next we extracted the street intersections to build a network in which nodes are intersections and directed links are roads connecting those intersections (we use directed links because a non-negligible fraction of streets in Manhattan are one-way). The extracted network of street intersections was then manually cleaned for obvious inconsistencies or redundancies (such as duplicate intersection points at the same geographic positions), in the end containing 4091 nodes and 9452 directed links. This network was used to map-match the GPS locations from the trip dataset. We only matched locations for which a closest node in the street intersection network exists with a distance less than than $100\,\mathrm{m}$. Finally, from the remaining 150 million trips we discarded about 2 million trips that had identical starting and end points, and trips that lasted less than one minute.

\subsection{Maximum matching of shareability networks}
Given a graph $G = (V,E)$, a matching $M$ in $G$ is a set of pairwise non-adjacent edges. A maximum matching is a matching that contains the largest possible number of edges. A weighted maximum matching is a matching in which the sum of edge weights is maximal. In the context of shareability networks, maximum matching solves optimizing the number of shared trips, while weighted maximum matching minimizes the cumulative time needed to accommodate all trips if the weights on the shareability network are taken as the travel time that is saved by sharing. Given that shareability networks are sparse,  for the case $k=2$ maximum matching and weighted maximum matching can be solved in polynomial times $O(n\sqrt{n})$ and $O(n^2 \log n)$ \cite{Galil86}, respectively, where $n$ is the number of nodes in the network. For higher dimensions, $k>2$, fast approximations to the optimal solutions exist \cite{Chandra01}, which however become computationally unfeasible for $k>3$. For details see SI Appendix.

\subsection{Subsampling of vehicles}
To assess to which extent our results could be generalized to cities with lower taxi densities than New York, or to situations where willingness to share is low, we have generated additional low density situations by subsampling our dataset, randomly removing various fractions of vehicles from the system in the following way: For each day in the dataset, we randomly selected a percentage $c$ of the taxis in the trace, and deleted the corresponding trips from the dataset. We varied $c$ from $95\%$ down to $1\%$, generating a number of trips per day as low as 1,962. Note that subsampling the vehicles we filter both taxis and the trips which represent the demand.

\begin{acknowledgments}
\subsection{ACKNOWLEDGMENTS}
The authors thank Chaogui Kang for GIS processing, as well as the Enel Foundation, the National Science Foundation, General Electric, the Rockefeller Foundation, the MIT SMART program, the MIT CCES program, Audi Volkswagen, BBVA, The Coca Cola Company, Ericsson, Expo 2015, Ferrovial, and all the members of the MIT Senseable City Lab Consortium for supporting this research.
\end{acknowledgments}
}


\end{article}

\end{document}


\date{}

\maketitle 


\noindent In this supporting information we present the detailed methods, including the handling of the data set, the formal derivation of the network-based approach used to quantify the benefits of a shared taxi system, and essential extended details of the analysis. 

\subsection*{Data set and pre-processing}
The data set contains origin-destination data of all 172 million trips with passengers of all 13,586 taxicabs in New York during the calendar year of 2011. Each vehicle is associated with a license, a so-called \emph{medallion}, which is synonymously used as a name for the vehicles. These medallion taxis are the only vehicles in the city permitted to pick up passengers in response to a street hail. A medallion may be purchased from the City at infrequent auctions, or from another medallion owner. Because of their high prices medallions and most cabs are owned by investment companies and are leased to drivers. There are 39,437 unique driver IDs in the data set, which corresponds to 2.9 drivers per medallion on average. Note that we unfortunately do not have explicit information on the number of passengers per vehicle, however, following the data reported by \cite{FactBook}, that the average daily number of passengers served by NY taxis is 600,000, with 450,000 trips on average, the average number of passengers per trip is around 1.3. The data set contains a number of fields from which we use the following: medallion ID, origin time, destination time, origin longitude, origin latitude, destination longitude, destination latitude. Times are accurate to the second, positional information has been collected via Global Positioning System (GPS) technology by the data provider. Out of our control are possible biases due to urban canyons which might have slightly distorted the GPS locations during the collection process \cite{grush2008cam}.
All IDs are given in anonymized form, origin and destination values refer to the origins and destinations of trips, respectively. 

For creating the street network of Manhattan we used data from openstreetmap.org. We filtered the streets of Manhattan, selecting only the following road classes: primary, secondary, tertiary, residential, unclassified, road, living street. Several other classes were deliberately left out, such as footpaths, trunks, links or service roads, as they are unlikely to contain delivery or pickup locations. Next we extracted the street intersections to build a network in which nodes are intersections and directed links are roads connecting those intersections (we use directed links because a non-negligible fraction of streets in Manhattan are one-way). The extracted network of street intersections was then manually cleaned for obvious inconsistencies or redundancies (such as duplicate intersection points at the same geographic positions), in the end containing 4091 nodes and 9452 directed links. This network was used to map-match the GPS locations from the trip data set. We only matched locations for which a closest node in the street intersection network exists with a distance less than than $100\,\mathrm{m}$. We matched GPS points to street intersections rather than to points on the closest street segments as a reasonable compromise between high accuracy (the average length of street segments in Manhattan is $126\,\mathrm{m}$) and granularity of discretization that is mainly relevant in the estimation of travel times, see SI Section \emph{Computing travel times}. Finally, from the remaining 150 million trips we discarded about 2 million trips that had identical starting and end points, and trips that lasted less than one minute.

\subsection*{Static and dynamic implementations, and the relevance of empty trips}
Before addressing the aim of this work of developing a theoretical framework for the rigorous quantification and optimization of general spatio-temporal sharing problems, in particular of trip sharing in taxi systems, we draw in this section the connection to the corresponding practical issues that come along with concrete implementations of such systems. A street hailing based, conventional taxi system, ideally features taxis that i) when empty, try to find a passenger as fast as possible by choosing an optimal passenger-finding strategy \cite{yamamoto2008adaptive,Li11}, ii) when occupied, deliver passengers as fast as possible to their destination. For taxi operations that are not based on street hailing but on trip queries, as facilitated by modern mobile phone apps and services, the sub-problem i) becomes an issue of efficiently locating, scheduling and dispatching the closest empty taxis. 

In conventional taxi systems, the problem of dispatching taxis to a trip request is conceptually straightforward: Find the closest empty taxi -- where the metric for closeness can involve vehicle velocities, traffic conditions \cite{yuan2011driving}, etc. In dynamic approaches to taxi sharing however, in which taxis are allowed to re-route and to pick up new passengers on the fly, the concrete operational issues of the detailed spatial query setup and the communication protocol design between taxis and dispatch system becomes potentially intricate. This set of problems has been satisfactorily solved by a recently published service model called T-share \cite{ma2013t}. Using empirical data and simulations, the study has pointed out the impressive potential of a heuristic, dynamic approach, and has provided efficient and scalable algorithms to solve the taxi searching and scheduling sub-problems which occur when taxis are allowed to change their routes on the fly. It found that a dynamically scheduled taxi system is able to handle ridesharing, allowing in the city of Beijing to service 25\% additional taxi users while saving 13\% travel distance compared to services without ridesharing. 

As we derive formally in the following sections, it is already possible to achieve high levels of trip sharing and extraordinary benefits with simple static implementations rather than with dynamic ones. A static perspective additionally allows rigorous quantification of the benefits using provably optimal algorithms from graph theory. Such a static implementation of taxi sharing can be imagined as follows, from the perspective of a passenger: 1) Submit trip request source and destination to the central system, for example via mobile phone app, 2) wait up to $\delta=1\,\mathrm{min}$ for the system to respond with a sharing option presenting the information of estimated arrival time without sharing and how, due to sharing, the trip might be prolonged by a time up to $\Delta$ (or equivalently, the passenger is presented a time window of possible arrival times), 3) either confirm or deny the presented sharing option. After the $k$ passengers have confirmed their participation, the shared trip becomes one assignment of bundled, unchangeable requests with a well-defined route having a starting point, intermediate points, and an end point, which can be handled by any conventional taxi dispatch system. Only at this step do the locations of the surrounding taxis become relevant. Therefore, a static implementation builds directly on existing taxi systems without the need to reconsider the dispatch process.

Since such a static implementation of a shared taxi system can be regarded as independent of the problem of dispatching and can build directly on top of existing dispatching solutions, it is not in the scope of our work to revisit implementation issues concerning search and scheduling of empty taxis. Still, it is helpful to understand quantitatively the relevance of empty trips in the general picture of taxi systems since one could suspect that the cruising of empty taxis might be the main source of wasting energy in terms of the effective mile per gallon. For this reason, we measured the distribution of durations of occupied trips versus the durations in-between occupied taxi trips, Fig.~S1. The figure shows the empirical probability distribution of durations of occupied trips, and of the time spans in-between the occupied trips which comprise both empty trips and all activities where taxis are not being used to transport passengers such as shift changes, lunch breaks, vehicle maintenance. Specific information for distinguishing between the empty trips and the rest of these downtime activities is unfortunately not available, as the data points only include the spatial and temporal information of the pickups and dropoffs of occupied trips. In any case, the in-between durations peak below two minutes, substantially lower than the durations of occupied trips which peak at around six minutes, showing that taxis tend to find new passengers relatively quickly and that taxis spend about 75\% of their on-service time performing occupied trips. This is at least the case for Manhattan; it is an open question whether the dispatching sub-problem becomes more relevant in cities with lower taxi demand. Further, the distribution of in-between durations is long-tailed due to the various types and occurences of taxi downtimes, see Fig.~S1 inset. For example, bumps in the distribution at 6 and 14 hours possibly indicate shift-related durations. However, the vast majority of empty trips are covered by the durations that are shorter than half an hour, which includes over $96\%$ of all downtimes, shown in the main panel of Fig.~S1.

For the specific problem of implementing a dynamic taxi dispatching and trip sharing system these observations suggest that while efficient dispatching and routing of empty taxis is doubtlessly a non-negligible issue, especially for the possible small fraction of cases where taxis do not find passengers quickly, it might be of higher importance to solve the main problem related to the occupied trips, namely of matching trips efficiently. In any case, to understand thoroughly a taxi system's improvability of the handling of empty trips, explicit data of empty trips must be available, preferably containing frequently sampled points per trip, which is not the case in the available data set.

\subsection*{A network-based approach for sharing taxi rides}

In contrast to typical approaches based on linear programs \cite{Berbeglia10,horn2002}
 and references therein, we show in this supplementary information in detail how our new approach allows poly\-nom\-ial-time, i.e., feasible, computation of the optimal ride sharing strategy when at most two trips can be combined, and polynomial-time computation of a constant-factor approximation of the optimal solution when $k>2$ trips can be shared. Notice, though, that the degree of the involved polynomials increases with $k$. In practice, the approach turns out to be computationally feasible for $k=3$, while it becomes impractical for larger values of $k$. 
 
The goal of the trip sharing strategy can be either minimizing the number of trips performed or the total travel cost for a given set of trips, subject to a quality of service constraint (maximum allowed delay at delivery/destination). The former goal allows quantifying the actual number of taxis needed to satisfy the current taxi demand with a shared taxi service. By assuming that cost of a trip is proportional to the travel time, the latter goal becomes a proxy of the carbon emissions generated by the shared taxi fleet to accommodate the total traffic demand. By comparing the total travel time of the shared taxi service with that of the traditional, non-shared taxi service, we can thus quantify the expected reduction in pollution achieved by a shared versus a traditional taxi service. Of course the factors which determine vehicle emissions can be complicated and highly non-linear, such as the most important factor of speed and engine load, which are themselves affected by traffic congestion, driver mentality, traffic signals, posted speed limits, etc. \cite{tong2000omv,kean2003evs}. However, {\em all things being equal}, as general traffic conditions presumably remain largely unaltered by the sharing service, the cumulative emissions can be treated as proportional to the travel time.

The high-level idea, elaborated rigorously in the following sections, is to cast the problem of identifying the best trip sharing strategy as a network problem, where nodes of the network represent taxi trips, and links connect trips that can be combined. The resulting network is called the {\em shareability network}. A maximum tolerated time delay $\Delta$ at both pickup and delivery location regulates the density of the shareability network -- the higher $\Delta$ the more sharing opportunities arise but the lower the quality of service becomes due to the increased delays. We show that the problem of finding the optimal trip sharing strategy when at most two trips can be combined is equivalent to the problem of finding the maximum matching in the shareability network, which can be solved in time $O(m\sqrt{n})$, where $n$ is the number of nodes and $m$ the number of links in the network.  Notice that the shareability network is likely sparse, i.e. the average node degree is a constant which does not depend on $n$, hence the above time complexity reduces to $O(n\sqrt{n})$. More specifically, the maximum matching in the shareability network corresponds to the trip combination strategy that minimizes the number of performed trips. If links in the trip graph are weighted with the travel cost reduction, i.e. the difference between the duration of combined ride and the two single rides, then the problem of finding the trip combination that minimizes the total travel cost is equivalent to the problem of finding the maximum weighted matching in the shareability network, which is also solvable in polynomial time.

If we relax the assumption that at most two trips can be combined, the complexity of the problem increases. In fact, the problem(s) at hand becomes equivalent to the weighted matching problem on $k$-bounded hyper-networks, which is NP-complete when $k>2$ in general hyper-networks. However, polynomial-time algorithms are known that compute a solution which is within a constant factor from optimal. In particular, when the number of combined trips is at most $k$, for any constant $k> 2$, simple greedy algorithms can be used to produce a solution within a factor $k$ from optimal.

We first present the case in which at most two trips can be combined, and then proceed to present the more general (and complex) case of an arbitrary number of combined trips. To ease presentation, we assume single passenger trips.

\subsection*{The two-trips sharing case}

In this section, we assume that at most two trips can be combined. Notice that this is a stricter condition than assuming that the maximum taxi capacity is two. The difference between the two assumptions is exemplified in Fig.~1G of the main text. We have three trips $T_1$, $T_2$, and $T_3$. Assuming that delay constraints on passenger delivery are satisfied, the three trips can be combined in a single trip even using a taxi with capacity two if the passenger of $T_2$ is loaded onboard taxi performing $T_1$ at time $t_2$, unloaded at time $t^{'}_2>t_2$, and the passenger of trip $T_3$ is loaded at time $t_3>t^{'}_2$ -- cfr. the middle case in Fig.~1G of the main text. This combination of trips is not allowed in our model, since only two trips at most can be combined. Notice, on the other hand, that any shared trip obtained by combining at most two single trips can be realized using a taxi with capacity two. More generally, any $k$ combination of trips can be performed using a taxi with capacity $k$. Hence, the trip combination solutions presented in the following can be accomplished using a taxi fleet where each taxi has capacity $k$, where $k$ is the upper bound on the number of trips combined in a single trip.

Let $S=(T,L)$ be the (undirected) {\em shareability network} defined as follows. The node set $T=\{T_1,\dots,T_n\}$ corresponds to the set of all possible $n$ trips. The link set $L=\{L_1,\dots,L_m\}$ is built as follows: link $(T_i,T_j)\in L$ connects nodes $T_i$ and $T_j$ if and only if trips $T_i$ and $T_j$ can be combined. The trip obtained combining trips $T_i$ and $T_j$ is denoted $T_{i,j}$ in the following. Whether trips $T_i,T_j$ can be combined depends on spatial/temporal properties of the two trips, and on an upper bound $\Delta$ on the maximum delivery delay customers can tolerate. How to derive the set of combinable trips (link set $L$) given a travel data set such as the one at hand is a problem we defer to the next section. In the remainder of this section, we assume that $L$ is known and given as input to the problem.
\begin{defn}
A set $\mathcal{T}$ of {\em possibly combined trips}, where combined trips are composed of two single trips, is defined as $\mathcal{T}=\mathcal{T}_1\cup\mathcal{T}_2$, where $\mathcal{T}_1\subseteq T$ is a set of single trips, and $\mathcal{T}_2$ is a set of combined trips $T_{i,j}$, for some $i,j\in\{1,\dots,n\}$.
\end{defn}

\begin{defn}[Feasible trip set]
A set $\mathcal{T}$ of possibly combined trips is {\em feasible} if and only if all trips in $T$ appear once in $\mathcal{T}$, formally: 
\[
\forall T_i\in T,~~~(T_i\in\mathcal{T}_1)\vee(\exists_1 j\mathrm{~s.t.~}T_{i,j}\in\mathcal{T}_2)~.
\]

\end{defn}

Notice that any trip $T_i\in T$ can appear only once in a feasible trip set, either as a single trip (if $T_i\in\mathcal{T}_1$), or combined with another trip (if $T_i\in\mathcal{T}_2$).
The two travel optimization problems we solve in the following are formally defined as follows:

\begin{defn}[{\sc MinimumNumberTrip -- MNT}]
Given the shareability network $S=(T,L)$, determine a feasible trip set of minimum cardinality.
\end{defn}

\begin{defn}[{\sc MinimizeTotalTravelCost -- MTTC}]\label{defWeight}
Given the weighted shareability network $S=(T,L)$ where each link $(T_i,T_j)\in L$ is weighted with $w_{ij}=c(T_i)+c(T_j)-c(T_{i,j})$, where $c(T_x)$ denotes the cost of trip $T_x$, and $c(T_{i,j})$ the cost of the combined trip $T_{i,j}$; determine a feasible trip set such that the total travel cost is minimized. 
\end{defn}

Regarding problem {\sc MTTC}, we observe that we can assume w.l.o.g. that $c(T_{i,j})< c(T_i)+c(T_j)$, since otherwise we can remove link $(T_i,T_j)$ from $L$ without impacting the optimal solution.

\begin{defn}{\sc (Matchings)}
Let $S=(T,L)$ be a shareability network. A matching on $S$ is a set of links $L'\subseteq L$ such that no two links in $L'$ share an endpoint. A maximum matching on $S$ is a matching on $S$ of maximum cardinality. If links in $S$ are weighted, a maximum weighted matching on $S$ is a matching $L'$ on $S$ such that the sum of weights of links in $L'$ is maximum.
\end{defn}

\begin{lemma}\label{equiv}
A set  $\mathcal{T}$ of possibly combined trips is feasible if and only if its subset $\mathcal{T}_2$ of combined trips is a matching of $S$.
\end{lemma}

\begin{proof}
Let $\mathcal{T}$ be any feasible set of combined trips. Since $\mathcal{T}$ is feasible, it contains either single trips, or trips obtained by combining two trips. Let $\mathcal{T}_2\subseteq\mathcal{T}$ be the set of combined trips in $\mathcal{T}$. For any combined trip $T_{i,j}\in\mathcal{T}_2$, we consider the corresponding link $(T_i,T_j)$ in the shareability network. Notice that, for any other link of the form $(T_i,T_h)$ (or $(T_h,T_j)$) in the shareability network, the corresponding combined trip $T_{i,h}$ (or $T_{h,j}$) is not in $\mathcal{T}_2$, since otherwise conditions on maximal trip combination would be violated. It follows that no two links corresponding to the combined trips in $\mathcal{T}_2$ share a node, i.e., the links corresponding to trips in $\mathcal{T}_2$ are a matching on $S$. 

\noindent The proof of the reverse implication is similar. Any matching $M$ of $S$ uniquely determines a set of combinable trips $\mathcal{T}_2$. A feasible set $\mathcal{T}$ of possibly combined trips is then obtained from $\mathcal{T}_2$ by adding as single trips all trips whose corresponding nodes in $S$ are not part of the matching $M$. 
\end{proof}

\begin{thm}\label{MaxMatching}
Let $M_{\mathrm{max}}\subseteq L$ be a maximum matching on $S$. Then, the minimum cardinality of a feasible set of possibly combined trips is $n-|M_{\mathrm{max}}|$.
\end{thm}
\begin{proof}
By Lemma \ref{equiv}, the cardinality of any feasible set $\mathcal{T}$ of possibly combined trips is $n-|M|$, where $|M|$ is the cardinality of the matching $M$ corresponding to the subset $\mathcal{T}_2$ of combined trips in $\mathcal{T}$. The proof then follows by observing that the minimum cardinality of a feasible set is obtained when $|M|$ is maximum, i.e., when $M$ is a maximum matching for $S$.
\end{proof}

The proof of Theorem \ref{MaxMatching} suggests a polynomial time algorithm for solving {\sc MNT}, which is reported below. The feasible set $\mathcal{T}$ of trips to be performed is initialized to the entire set of single trips $T$. Given the shareability network $S$, a maximum matching $M_{\mathrm{max}}$ on $S$ is computed using, e.g., Edmond's algorithm \cite{Galil86}.
 For any edge $(T_i,T_j)\in M_{\mathrm{max}}$, the combined trip $T_{i,j}$ is included in $\mathcal{T}$, while the individual trips $T_i$ and $T_j$ are removed from $\mathcal{T}$. After all edges in $M_{\mathrm{max}}$ have been considered and processed as above, $\mathcal{T}$ contains a set of (possibly combined) trips of minimum size that satisfies all customers, i.e., it is a feasible trip set of minimum size. The time complexity of {\sc MaxMatch} is determined by the complexity of the matching algorithm. Considering that the shareability network is likely to be very sparse in practice, the Edmond's matching algorithm yields $O(n\sqrt{n})$ time complexity.

\begin{center}
{\begin{minipage}{120mm}\hrulefill
\small
\begin{tabbing}
Alg\=or\=ithm\=~ {\sc MaxMatch}\\
{\em Input}:\>\>\>the shareability network $S=(T,L)$\\
{\em Output}: \>\>\>the set $\mathcal{T}$ of (possibly combined) trips to be performed\\
1.\> $\mathcal{T}=T$\\
2.\> Build a maximum matching $M_{\mathrm{max}}$ on $S$\\
3.\> {\bf for each} $(T_i,T_j)\in M_{\mathrm{max}}$ {\bf do}\\
4.\>\>$\mathcal{T}=\mathcal{T}\cup\{T_{i,j}\}$; $\mathcal{T}=\mathcal{T}-\{T_i,T_j\}$\\
5\> {\bf return} $\mathcal{T}$
\end{tabbing}
\vspace*{-3mm}
\hrulefill
\end{minipage}}\\[10pt]
Algorithm {\sc MaxMatch} for optimally solving {\sc MNT}.
\end{center}

\begin{thm}\label{MaxWeightMatching}
Let $M_{\mathrm{max}_w}\subseteq L$ be a maximum weighted matching on $S$, where $S$ is link weighted as described in Definition \ref{defWeight}. Then, the feasible set of possibly combined trips of minimum total travel cost has cost
 \[
 c_{\mathrm{min}}=\sum_{i=1,\dots,n}{c(T_i)}-\sum_{(T_i,T_j)\in M_{\mathrm{max}_w}} c(T_i)+c(T_j)-c(T_{i,j})~.
 \]
\end{thm}
\begin{proof}
By Lemma \ref{equiv}, the subset $\mathcal{T}_2$ of combined trips of any feasible trip set $\mathcal{T}$ corresponds to a matching $M$ on $S$. For any edge $(T_i,T_j)\in M$, the travel cost reduction due to the combined trip $T_{i,j}$ with respect to the cost of the two single trips $T_i,T_j$ is $c(T_i)+c(T_j)-c(T_{i,j})$. Thus, the total travel cost for any feasible trip set $\mathcal{T}$ is given by the total travel cost of the single trips ($\sum_{i=1,\dots,n}{c(T_i)}$), minus the sum of the cost savings achieved by the combined trips in $M$, i.e., $\sum_{(T_i,T_j)\in M} c(T_i)+c(T_j)-c(T_{i,j})=\sum_{(T_i,T_j)\in M} w_{ij}$. The proof then follows by observing that, if $M_{\mathrm{max}_w}$ is a maximum weighted matching on $S$, then the sum of the cost savings is maximized, and the total travel cost of $\mathcal{T}$ is minimized.
\end{proof}

The algorithm {\sc WeightedMaxMatch} to find the feasible trip set of minimum travel cost is similar to {\sc MaxMatch}, the only difference being that the maximum matching algorithm in step 2 is replaced with a maximum weighted matching algorithm. For instance, we can use Edmond's algorithm for weighted matching \cite{Galil86}
, which on a sparse graph yields time complexity $O(n^2 \log n)$.

\subsubsection*{The $k$-trips sharing case}

In this section, we generalize the results presented in the previous section to the more challenging scenario in which an arbitrary number $k> 2$ of trips can be combined. The only assumption we make about the value of $k$ in this section is that $k$ does not depend on the total number of trips $n$, i.e., $k=O(1)$ in asymptotic notation. Considering that $k$ is also an upper bound on taxi capacity needed to accommodate the $k$ combined trips, assuming $k$ a small constant is reasonable in any practical case.

We first present how a combination of up to $k$ trips can be represented in form of a $k$-bounded shareability hyper-network. Some definitions are in order before proceeding further.

\begin{defn}
A set $\mathcal{T}$ of {\em possibly combined trips}, where combined trips are composed of at most $k\ge 2$ single trips, is defined as $\mathcal{T}=\mathcal{T}_1\cup\dots\cup\mathcal{T}_k$, where $\mathcal{T}_1\subseteq T$ is a set of single trips, and $\mathcal{T}_h$, with $2\le h\le k$, is a set of combined trips $T_{i_1,\dots,i_h}$, for some $i_1,\dots,i_h\in\{1,\dots,n\}$.
\end{defn}

\begin{defn}[Feasible trip set]
A set $\mathcal{T}$ of possibly combined trips is {\em feasible} if and only if all trips in $T$ appear once in $\mathcal{T}$, formally: 
\[
\forall T_j\in T,~~~(T_j\in\mathcal{T}_1)\vee(\exists_1 h,\ell\mathrm{~s.t.~}(T_{i_1,\dots,i_h}\in\mathcal{T}_h)\wedge(i_\ell=j))~.
\]

\end{defn}

Notice that also in this case, any trip $T_i\in T$ can appear only once in a feasible trip set, either as a single trip, or in a combined trip.

\begin{defn}
A {\em hyper-network} $H$ is a pair $H=(T,\mathcal{L})$ where $T=(T_1,\dots,T_n)$ is a set of nodes (representing single trips in the context at hand), and $\mathcal{L}$ is a set of non-empty subsets of $T$ called hyper-links. The {\em size} of a hyper-link is the number of nodes it connects. A hyper-network whose hyper-links have size $\le k$, for some integer $k\ge 2$, is called a $k${\em -bounded hyper-network}.
\end{defn} 

\begin{defn}
A (hyper-)matching $M$ on the hyper-network $H=(T,\mathcal{L})$ is a subset of the hyper-links in $\mathcal{L}$ such that each node in $T$ appears in at most one hyper-link.
\end{defn}

Similarly to the previous section, given a set of trips $T$, we can represent all possible combinations of up to $k$ trips -- defined according to some quality of service criterion -- with a $k$-bounded hyper-network, which we call the {\em shareability hyper-network}. Formally, the trip hyper-newtork is defined as the $k$-bounded hyper-network $H=(T,\mathcal{L})$, where $L_i=(T_{i_1},T_{i_2},\dots)\in \mathcal{L}$ if and only if trips $T_{i_1},T_{i_2},\dots$ can be combined. Notice that, if hyper-link $L_i=(T_{i_1},T_{i_2},\dots)$ belongs to the shareability hyper-network, so do all hyper-links formed of any subset of the nodes connected by $L_i$. This is due to the fact that, if, say, trip $T_{1,2,3,4}$ is feasible, so do trips $T_{1,2}$, $T_{1,2,3}$, etc. We call a hyper-link in $H$ {\em maximal} if its incident nodes are not a subset of any other hyper-link in $H$.

We are now ready to formally define the two considered optimization problems.

\begin{defn}[$k${\sc  -MinimumNumberTrip} -- $k${\sc MNT}]
Given the shareability hyper-network $H=(T,\mathcal{L})$, determine a feasible trip set $\mathcal{T}$ of minimum cardinality.
\end{defn}

\begin{defn}[$k${\sc -MinimizeTotalTravelCost} -- $k${\sc MTTC}]
Given the weighted shareability hyper-network $H=(T,\mathcal{L})$ where each link $L_i=(T_{i_1},T_{i_2},\dots)\in \mathcal{L}$ is weighted with $w_i^c=\sum_{i_j\in L_i} c(T_{i_j})-c(T_{i_1,i_2,\dots})$, where $c(T_{i_j})$ denotes the cost of trip $T_{i_j}$, and $c(T_{i_1,i_2,\dots})$ the cost of the combined trip $T_{i_1,i_2,\dots}$; determine a feasible trip set $\mathcal{T}$ such that the total travel cost is minimized. 
\end{defn}

\begin{lemma}\label{equivK}
A set  $\mathcal{T}$ of possibly combined trips is feasible if and only if its subset $\mathcal{T}_c=\mathcal{T}-\mathcal{T}_1$ of combined trips is a (hyper-)matching of $H$.
\end{lemma}

\begin{proof}
The proof is along the same lines of the proof of Lemma \ref{equiv}.
\end{proof}

\begin{thm}\label{kWeightMatch}
Let $H=(T,\mathcal{L})$ be a shareability hyper-network, and assign weight $w_i=|L_i|-1$ to each hyper-link $L_i\in \mathcal{L}$. Then, the minimum cardinality of a feasible set of possibly combined trips is $n-\sum_{L_i\in M_{\mathrm{max}_w}} w_i$, where $M_{\mathrm{max}_w}$ is a maximum weighted matching of $H$.
\end{thm}
\begin{proof}
By Lemma \ref{equivK}, any feasible trip set $\mathcal{T}$ uniquely defines a matching $M$ in the shareability hyper-network $H$. Consider any hyper-link $L_i$ in the matching $M$. By definition, $w_i$ represents the number of trips that are saved by performing the combined trip corresponding to hyper-link $L_i$ instead of performing all single trips in $L_i$. For instance, if $L_i=\{T_{i_1},\dots,T_{i_k}\}$, the combination of $k$ trips allows reducing the number of performed trips from $k$ to 1; i.e., the total number of trips is reduced of $w_i=k-1$. Based on this observation, the total number of trips performed for feasible trip set $\mathcal{T}$ equals $n-\sum_{L_i\in M} w_i$. The proof then follows by observing that the total number of trips is minimized when $\sum_{L_i\in M} w_i$ is maximized, i.e., when $M$ is a maximum weighted matching for $H$.
\end{proof}

Unfortunately, the maximum (weighted) matching problem on $k$-bounded hyper-networks is NP-complete for $k>2$ on general hyper-networks, hence finding the optimal solution to $k${\sc MNT} is likely computationally hard. However, a simple greedy heuristic can be used to find a $k$-approximation of the optimal solution in time $O(m \log m)$, where $m$ is the number of hyper-links in the hyper-network, which yields $O(n\log n)$ complexity under our working assumption of sparse shareability hyper-network. In the greedy heuristic, a hyper-link $L_i$ of maximum weight is added to the current matching at each iteration, and hyper-links sharing at least one node with $L_i$ are removed from the set of candidate hyper-links for matching before proceeding to the next iteration. Observe that better approximation ratios can be obtained at the price of increased (but still polynomial) time complexity using, for instance, the algorithm of \cite{Chandra01}
 which finds a $2(k+1)/3$ approximation of the optimal solution. Note that the weighted maximum matching problem on hyper-networks is equivalent to the weighted set packing problem. The greedy algorithm for finding a $k$-approximation to the optimal $k${\sc MNT} solution is reported below.

\begin{center}
{\begin{minipage}{120mm}\hrulefill
\small
\begin{tabbing}
Alg\=or\=ithm\=~ {\sc GreedyKMatching}\\
{\em Input}:\>\>\>the shareability hyper-network $H=(T,\mathcal{L})$ with weights $w_i$ on hyper-links\\
{\em Output}: \>\>\>the set $\mathcal{T}$ of (possibly combined) trips to be performed\\
1.\> $\mathcal{T}=T$\\
2.\> Build a weighted matching $M_{w}$ of $H$ using the greedy heuristic\\
3.\> {\bf for each} $L_i=(T_{i_1},\dots,T_{i_j})\in M_{w}$ {\bf do}\\
4.\>\>$\mathcal{T}=\mathcal{T}\cup\{T_{i_1,\dots,i_j}\}$; $\mathcal{T}=\mathcal{T}-\{T_{i_1}\}-\dots-\{T_{i_j}\}$\\
5\> {\bf return} $\mathcal{T}$
\end{tabbing}
\vspace*{-3mm}
\hrulefill
\end{minipage}}\\[10pt]
Algorithm {\sc GreedyKMatching} for finding a $k$-approximation to $k${\sc MNT}.
\end{center}

\begin{thm}

Let $H=(T,\mathcal{L})$ be the shareability hyper-network, where each $L_i=(T_{i_1},T_{i_2},\dots)\in\mathcal{L}$ is weighted with the weight $w^c_i=\sum_{i_j\in L_i}c(T_{i_j})-c(T_{i_1,i_2,\dots})$ representing the cost saving in performing the combined trip versus the collection of single trips. Then, the feasible set of possibly combined trips of minimum total travel cost has cost
\[
 c_{\mathrm{min}}=\sum_{i=1,\dots,n}{c(T_i)}-\sum_{L_i\in M_{\mathrm{max}_w}} w_i^c~,
 \]
where $M_{\mathrm{max}_w}$ is a maximum weighted matching of $H$.
\end{thm}
\begin{proof}
The proof is along the same lines of the proof of Theorem \ref{kWeightMatch}. 
\end{proof}

The greedy heuristic for computing a $k$ approximation of the optimal solution to $k${\sc MTTC} can be straightforwardly obtained from Algorithm {\sc GreedyKMatching} by using weights $w_i^c$ instead of $w_i$ to label hyper-links in the shareability hyper-network. 

\subsection*{Building the shareability network}

In this section, we describe a method for producing the shareability (hyper-)network, given a set of single trips $T=\{T_1,\dots,T_n\}$ and a quality of service criterion $\Delta$. We present in detail the method for $k=2$, and shortly describe how the technique can be generalized to arbitrary values of $k$.

Each trip $T_i\in T$ is characterized by the following quantities: the trip origin $o_i$ and destination $d_i$, that we can think of as pairs of $(lat,lon)$ coordinates; the start time $st_i$; and the arrival time $at_i$. We start defining a notion of feasible trip combination based on a quality of service criterion $\Delta$.

\begin{defn}
The combined trip $T_{i,j}$ is {\em feasible} if and only if a trip route can be found such that the following conditions are satisfied:
\begin{itemize}
\item[a)] $st_i\le pt_i\le st_i+\Delta$;
\item[b)] $st_j\le pt_j\le st_j+\Delta$;
\item[c)] $dt_i\le at_i+\Delta$;
\item[d)] $dt_j\le at_j+\Delta$;
\end{itemize}
where $pt_x$ is the pickup time at $o_x$ in the combined trip, and $dt_x$ is the delivery time at $d_x$ in the combined trip.
\end{defn}
The above definition is motivated by the fact that a customer might be willing to wait at most some extra time $\Delta$ at her pickup location (and in general she might not be able to show up at $o_i$ before time $st_i$), as well as to arrive at destination with delay at most $\Delta$ (early arrivals are likely not to be a problem for customers).

\begin{thm}\label{buildTripGraph}
Building the shareability network $S=(T,L)$ starting from the trip set $T=\{T_1,\dots,T_n\}$ requires $O(n^2)$ time.
\end{thm}
\begin{proof}
In the worst-case, we have to consider all $O(n^2)$ possible pairs of trips $T_i,T_j$. For each pair, the feasibility condition for the combined trip $T_{i,j}$ can be verified in $O(1)$ as follows. Observe that only four routes are possible for trip $T_{i,j}$: $o_i\to o_j\to d_i\to d_j$, $o_i\to o_j\to d_j\to d_i$, $o_j\to o_i\to d_i\to d_j$, and $o_j\to o_i\to d_j\to d_i$. Let us consider a specific route, e.g., $o_i\to o_j\to d_i\to d_j$. Condition $a)$ for feasibility is always satisfied by setting a pickup time at $o_i$ in the desired time window. The pickup time at $o_j$ can then be computed as follows: $pt_j= pt_i+tt(o_i,o_j)$, where $tt(x,y)$ denotes the travel time between $x$ and $y$. The delivery time at $d_i$ is defined as follows: $dt_i=pt_j+tt(o_j,d_i)$. Finally, the delivery time at $d_j$ is defined as $dt_j=dt_i+tt(d_i,d_j)$. Thus, the feasibility condition for route $o_i\to o_j\to d_i\to d_j$ can be verified by checking whether a value of $pt_i$ that simultaneously satisfies the four conditions below exists, which requires $O(1)$ time:
\begin{eqnarray}
st_i \le pt_i&\le& st_i+\Delta\\
st_j\le pt_i+tt(o_i,o_j)&\le& st_j+\Delta\\
pt_i+tt(o_i,o_j)+tt(o_j,d_i)&\le& at_i+\Delta\\
pt_i+tt(o_i,o_j)+tt(o_j,d_i)+tt(d_i,d_j)&\le& at_j+\Delta
\end{eqnarray}
The feasibility conditions for the other routes can be verified similarly. If there exists at least one route which satisfies the feasibility condition, then trip $T_{i,j}$ is feasible, and link $(T_i,T_j)$ is included in the trip graph. Otherwise, trips $T_i$ and $T_j$ cannot be combined.

\noindent Observe that the number of trip pairs to consider for combination can be reduced by considering only trip pairs $T_i,T_j$ such that: $i)$ $st_j\le at_i+\Delta$; and $ii)$ $st_i\le at_j+\Delta$. Simultaneously satisfying conditions $i)$ and $ii)$ is a necessary (but not sufficient) condition for feasibility of trip $T_{i,j}$. In practice, this heuristic considerably reduces the running time of the shareability network construction algorithm, although the worst-case time complexity remains $O(n^2)$.
\end{proof}

The algorithm for building the trip graph reported in the proof of Theorem \ref{buildTripGraph} can be extended in a straightforward way to the case of combinations of up to $k$ trips, yielding a time complexity of $O(n^k)$; in particular, all possible routes connecting $k$ origins with $k$ destinations, subject to the condition that each origin must precede the respective destination in the route, must be considered. The number of possible such routes grows exponentially with $k$, which is however assumed to be a small constant in our model. For instance, 60 possible routes connecting origins with destinations must be considered when $k=3$.

\subsection*{Runtime and feasibility considerations}
Although we have proven above theoretically that the trip sharing problem can be solved by maximum matching algorithms with polynomial time complexity, it is important to demonstrate that the computational implementation indeed runs fast enough on practical problem sizes and that the solution is feasible to operate in an actual real-time taxi sharing system. On average, in every two minutes there are about 600 trips requested in NYC, which constitutes the typical number of nodes in the shareability network that have to be matched optimally. To account for possible extraordinary spikes in the demand, we measured the runtime for sharebility networks of size 10,000. Here, the runtime of the software implemented in the programming language C on a Linux-based workstation equipped with i7-3930K CPU and 32GB of RAM is less than 0.1 seconds. For the more realistic size of 1,000 nodes the runtime is even less than 0.01 seconds. These measurements show that the proposed graph matching-based solution is highly efficient and feasible for an online taxi sharing service. Further, the procedures for inserting/deleting nodes and corresponding links in the shareability network are also very efficient, with a running time that increases only linearly with the number of nodes in the network.

\subsection*{Oracle and Online model}

Trip sharing opportunities are investigated according to two different models, called the {\em Oracle} and the {\em Online} model. The difference between the two models lies in the set of links in the shareability network that are considered for computing the maximum matching. In the Oracle model, all possible links between shareable trips as determined by spatial and temporal conditions are retained. Thus, two trips $T_1, T_2$ are considered to be shareable also if their starting times are separated by a long time interval (say, $30~\mathrm{min}$), as long as the two trips can be combined without imposing delays exceeding $\Delta$ on both trips. Note that this is in principle possible if, say, $T_1$ is a very long trip, and the pickup location of $T_2$ is close to the trajectory of $T_1$. Since in this paper we consider only static trip sharing, combining trips in a situation like the one described above would require the presence a reservation system, in which requests for taxi trips are issued well ahead of time (e.g., at least $30~\mathrm{min}$ in the example above). Since in many cases a reservation system is not present or not allowed by regulations (as it is the case in the city of NY), we have introduced also the Online model, which can be easily turned into a practical real time, on-demand taxi system. The idea is to reduce the number of links in the shareability network, by filtering out trip sharing opportunities for trips $T_1,T_2$ whose starting times are more than $\delta$ apart. In other words, we retain in the shareability network only trip sharing opportunities that can be exploited with a real-time, on-demand taxi system.

\subsection*{Computing travel times}
Knowledge of estimated travel times between arbitrary origin/destination in the road map is a pre-requisite for checking the trip sharing conditions, and, hence, to build the shareability network. Since we cannot use directly GPS taxi traces for this purpose due to lack of trajectory and speed information in the data set, we designed a travel time estimation heuristic starting from the pickup/drop off times recorded in the data set.

Given is a set of actually performed trips $\mathcal{T}=\{T_1,\dots,T_k\}$, where each trip $T_i=(o_i,d_i,tt_i)$ is defined by an origin location $o_i$, a destination location $d_i$, and a travel time $tt_i$. While in the original data set origin and destination of a trip are defined as raw GPS $(lat,lon)$ coordinates, in the following we assume origin and destination of a trip are taken from the set $\mathcal{I}$ of street intersections in the road map. To convert raw GPS coordinates into an intersection in $\mathcal{I}$, we associate $o_i$ (or $d_i$) to the closest intersection based on geodesic distance, subject to the condition that the distance to the closest intersection is below a threshold such as twice the average GPS accuracy, set to $100\,\mathrm{m}$. Thus, in the following we assume $o_i$ and $d_i$ are indeed distinct elements of the set $\mathcal{I}$ of possible intersections in the road map, i.e., $\forall T_i\in\mathcal{T},o_i,d_i\in\mathcal{I}$. We also define the set $\mathcal{S}=\{S_1,\dots,S_h\}$ of {\em streets} as the set of all road segments connecting two adjacent intersections in the road map.

Given the trip set $\mathcal{T}$ as defined above, the problem to solve is estimating the travel time $x_i$ for each street $S_i\in\mathcal{S}$, in such a way that the average relative error (computed across all trips) between the actual travel time $tt_i$ and the estimated travel time $et_i$ for trip $T_i$ computed starting from the $x_i$s (compound with a routing algorithm) is minimized\footnote{Formally, the average relative error is defined as $\epsilon_i=|tt_i-te_i|/tt_i$.}. Once error minimizing travel times for each street in $\mathcal{S}$ are determined, the travel time between any two intersections $I_i, I_j\in\mathcal{I}$ can be computed starting from the $x_i$s, using a routing algorithm that minimizes the travel time between any two intersections. Besides the trip set $\mathcal{T}$, we are also given the array $Le=(l_i)$ of the lengths of the streets in $\mathcal{S}$. 

In the following, we define the problem at hand more formally. First, we partition the trip set in time sliced subsets $\mathcal{T}_1,\dots,\mathcal{T}_{24}$, where subset $\mathcal{T}_i$ contains all trips whose starting time is in hour $i$ of the day. Finer partitioning (e.g., per hour and weekday, per hour and weekday and month, etc.) is possible, if needed. In the following, to simplify notation, we re-define $\mathcal{T}$ as any of the time-sliced subsets $\mathcal{T}_i$. In fact, the travel time estimation process can be performed independently on each of the time-sliced trip subsets. When a time-sliced trip set $\mathcal{T}$ is considered, classes $\mathcal{T}^1,\dots,\mathcal{T}^h$ of {\em equivalent trips} are formed, where two trips $T_u,T_v$ are said to be equivalent if and only if $(o_u=o_v)\wedge(d_u=d_v)$. Notice that, under the assumption that the routing algorithm is deterministic (i.e., it always computes the same route given the same starting and ending intersections $I_i$ and $I_j$), the set of streets in the optimal route from origin to destination is the same for any two trips $T_u,T_v\in\mathcal{T}^{i,j}$, where $\mathcal{T}^{i,j}$ is the class of trips with origin $I_i$ and destination $I_j$. Thus, all the trips in equivalence class $\mathcal{T}^{i,j}$ can be considered as multiple samples of the travel time on the same set of streets. All trips in $\mathcal{T}^{i,j}$ are then replaced by a single trip $T_{i,j}$ with corresponding origin and destination, and travel time $\bar{tt}_{i,j}$ equal to the average of the travel times of all trips in $\mathcal{T}^{i,j}$. After this step is performed for all equivalence classes, we are left with an aggregate set $\mathcal{T}_{\mathrm{agg}}$ of singleton, non-equivalent trips $T_{i,j}$, and corresponding travel times $\bar{tt}_{i,j}$. 

The travel time estimation heuristic is reported below. Initially, trips are filtered to remove ``loop" trips (i.e., trips with the same origin and destination), as well as excessively ``short" or ``long" trips. After a step in which initial routes are computed using a pre-selected initial speed $v_{\mathrm{init}}$ (the same for all streets), a second trip filtering step is performed, in which  excessively ``fast" and ``slow" trips are removed from the travel time estimation process. The rationale for this filtering is removing ``noisy" data which could have been resulted from very specific conditions (say, a snowstorm could have caused many slow trips). Including ``noisy" data in the travel time estimation process would bias the estimation process to partially compensate for ``noisy" trips, increasing the error experienced in the remaining portion of trips.

An iterative process is started after the second trip filtering step. The iterative process is composed of two nested iterations. In the outer iteration, new routes for the trips are computed based on the updated travel time estimation of street segments. After routes are computed, new trip travel time estimations are determined, and the average relative error across all trips is computed. Furthermore, an offset value is computed for each street segment, indicating whether travel times of all trips in which a street segment is included are under- or over-estimated. Then, an inner loop is started, with the purpose of refining street travel time estimations based on the computed offset values: an increase/decrease step {\tt k} is initialized, and used to tentatively change street travel time estimates based on the offset value (tentative updated trip travel times are accepted only if the resulting average speed $v$ on the trip is such that $0.5\,\mathrm{m}/\mathrm{sec}<v<30\,\mathrm{m}/\mathrm{sec}$). The tentative estimations are accepted if the newly computed average relative error is decreased with respect to the current value. Otherwise, another iteration of the inner loop is started with a smaller value of {\tt k}. This process is repeated until either the street travel time estimations are updated, or the value of {\tt k} has reached a specified minimum value. The outer iterative process terminates when there is no updated street travel time estimation after the execution of the inner loop.

After the iterative process, the algorithm has produced a travel time estimation for each street included in at least one optimal route for at least one trip in the data set (set $\mathcal{S}_{\mathrm{trip}}$).
The travel time for the remaining streets is then computed according to a simple heuristic: the travel time for each street $s$ having an intersection in common with at least one street in $\mathcal{S}_{\mathrm{trip}}$ is estimated based on the average speed estimated in adjacent streets, i.e., streets $s'$ such that $s$ and $s'$ share an intersection. This process is repeated until the travel time on all streets can be estimated. Finally, at step 7 the travel time between any two possible intersections $I_i,I_j$ in the street map is computed by first computing the optimal route between $I_i$ and $I_j$ using Dijkstra algorithm with the estimated trip travel times, and then computing the travel time by summing up the travel time of the streets in the optimal route. Notice that we use the Dijkstra algorithm \cite{Cormen90}
 (repeated $|\mathcal{I}|^2$ times) to compute all-to-all shortest paths instead of the classical Floyd-Warshall algorithm since the graph corresponding to the street network is very sparse. Thus, repeating  $|\mathcal{I}|^2$ times the Dijkstra algorithm yields $O(|\mathcal{I}|^2\log |\mathcal{I}|)$, which is lower than the $O(|\mathcal{I}|^3)$ complexity of Floyd-Warshall.

The travel time estimation algorithm has been executed on the set of about 150 millions trips performed in New York City during weekdays, in year 2011. The performance of the travel time estimation algorithms for the 24 trip classes (corresponding to time of day) is summarized in Table~S1. The table reports the average relative error computed on all trips retained after the filtering steps, the percentage of trips retained in the data set after filtering, and the number of streets included in at least one optimal route. As seen from the table, the algorithm provides travel time estimations incurring an average relative error of $15\%$. The vast majority of trips is retained in the data set after filtering (more than $97\%$ on the average). Furthermore, the vast majority of street segments are included in at least one optimal route: considering that the total number of (directed) street segments in Manhattan is 9452, on  average $91.7\%$ of the streets are included in at least one optimal route. For the remaining streets, step 6 of the algorithm is used to estimate street travel time.

To study the travel speeds estimated by our algorithm we calculated travel speeds across different times of the day. The travel time estimations are reasonable, with a relatively lower average speed of around $5.5\,\mathrm{m/sec}$ estimated during rush hours (between 8am and 3pm), and peaks around $8.5\,\mathrm{m/sec}$ at midnight. Further evidence for a reasonable estimation is highlighted also by Fig.~S2, reporting the estimated travel speed on each street segment at four different times of day: 0am, 8am, 4pm, and 22pm. As expected, travel speeds tend to reduce during daytime. Also, the algorithm is able to faithfully model the higher speed on highways (on the left-hand side of Manhattan).

\begin{center}
{\begin{minipage}{170mm}\hrulefill
\tiny
\begin{tabbing}
Alg\=or\=ithm\=~fo\=r tr\=avel time estimation\\
{\em Input}:\>\>\>the (sub)set $\mathcal{T}$ of performed trips;the set $\mathcal{I}$ of intersections; \\
\>\>\> the set $\mathcal{S}$ of streets; the vector $Le$ of lengths for streets in $\mathcal{S}$\\
{\em Output}: \>\>\>a travel time estimation matrix $ET(i,j)$, where $et_{i,j}$ is the estimated\\
\>\>\>time for traveling from intersection $I_i$ to intersection $I_j$\\
1.\> {\bf Equivalent trip reduction}\\
\>\> - group in class $\mathcal{T}^{i,j}$ all trips $T_u$ such that $o_u=I_i$ and $d_u=I_j$\\
\>\>- for each class $\mathcal{T}^{i,j}$, replace all trips in $\mathcal{T}^{i,j}$ with a single trip $T_{i,j}$ with $o_{i,j}=I_i$, $d_{i,j}=I_j$,\\
\>\> and $tt_{i,j}=\bar{tt}_{i,j}=\frac{\sum_{T_u\in\mathcal{T}^{i,j}}tt_u}{|\mathcal{T}^{i,j}|}$\\
\>\>- let $\mathcal{T}_{\mathrm{agg}}$ be the collection of trips $T_{i,j}$\\ 
2.\> {\bf First trip filtering}\\
\>\> - for each $T_{i,j}\in \mathcal{T}_{\mathrm{agg}}$, remove $T_{i,j}$ from $\mathcal{T}_{\mathrm{agg}}$ if $i=j$ //{\em remove ``loop" trips}\\
\>\> - for each $T_{i,j}\in \mathcal{T}_{\mathrm{agg}}$, remove $T_{i,j}$ from $\mathcal{T}_{\mathrm{agg}}$ if $(\bar{tt}_{i,j}<2min)$ or $(\bar{tt}_{i,j}>1h)$ //{\em remove ``short" and ``long" trips}\\
3.\> {\bf Initial route computation}\\
\>\> - for each $S\in\mathcal{S}$, set same initial speed $v_S=v_{\mathrm{init}}$; set travel time to $t_S=\frac{L(S)}{v_S}$\\
\>\> - for each $T_{i,j}\in\mathcal{T}_{\mathrm{agg}}$, compute optimal route $I_i\to I_j$ using Dijkstra algorithm\\
\>\> - for each $T_{i,j}\in\mathcal{T}_{\mathrm{agg}}$, let $\mathcal{S}^{i,j}=\{S^{i,j}_1,\dots,S^{i,j}_k\}$ be the set of streets in the optimal route for $T_{i,j}$\\
4.\> {\bf Second trip filtering}\\
\>\> - for each $T_{i,j}\in \mathcal{T}_{\mathrm{agg}}$, compute the average speed $as_{i,j}=\frac{\sum_h L(S_h^{i,j})}{\bar{tt}_{i,j}}$\\ 
\>\> - for each $T_{i,j}\in \mathcal{T}_{\mathrm{agg}}$, remove $T_{i,j}$ from $\mathcal{T}_{\mathrm{agg}}$ if $(as_{i,j}<0.5\,\mathrm{m}/\mathrm{sec})$ or $(as_{i,j}>30\,\mathrm{m}/\mathrm{sec})$ //{\em remove ``slow" and ``fast" trips}\\
5.\> {\bf Iterative steps}\\
\>\>5.1 set {\tt again}={\bf true}\\
\>\>5.2 {\bf while} {\tt again} {\bf do}\\
\>\>\> - {\tt again}={\bf false}\\
\>\>\>- for each $T_{i,j}\in\mathcal{T}_{\mathrm{agg}}$, compute optimal route $I_i\to I_j$ using Dijkstra algorithm\\
\>\>\>- for each $T_{i,j}\in\mathcal{T}_{\mathrm{agg}}$, let $\mathcal{S}^{i,j}=\{S^{i,j}_1,\dots,S^{i,j}_k\}$ be the set of streets in the optimal route for $T_{i,j}$\\
\>\>\>- for each $T_{i,j}\in\mathcal{T}_{\mathrm{agg}}$, compute $et_{i,j}=\sum_{S\in\mathcal{S}^{i,j}}t_S$ //{\em trip travel time estimation}\\
\>\>\>- let $\mathcal{S}_{\mathrm{trip}}=\bigcup_{\mathcal{T}^{i,j}\in\mathcal{T}_{\mathrm{agg}}}\mathcal{S}^{i,j}$\\
\>\>\>- {\tt RelErr}=$\sum_{T_{i,j}\in\mathcal{T}_{\mathrm{agg}}}\frac{|et_{i,j}-\bar{tt}_{i,j}|}{\bar{tt}_{i,j}}$\\
\>\>\>- for each $S\in \mathcal{S}_{\mathrm{trip}}$, let $T_S=\{T_{i,j}\in\mathcal{T}_{\mathrm{agg}}|S\in\mathcal{S}^{i,j}\}$ //{\em set of trips including} $S$ {\em in the current route}\\
\>\>\> - for each $S\in \mathcal{S}_{\mathrm{trip}}$, compute $O_S=\sum_{T_{i,j}\in T_S}(et_{i,j}-\bar{tt}_{i,j})\cdot |\mathcal{T}_{i,j}|$ //{\em offset computation}\\
\>\>\> - {\tt k}=1.2\\
\>\>\> 5.3 {\bf while true do}\\
\>\>\>\> - for each $S\in \mathcal{S}_{\mathrm{trip}}$, do the following\\
\>\>\>\>\> - if $O_S<0$, then $t_S=k\cdot t_S$; else $t_S=\frac{t_S}{k}$ //{\em street travel time estimate is increased/reduced based on offset}\\
\>\>\>\> - for each $T_{i,j}\in\mathcal{T}_{\mathrm{agg}}$, compute $et'_{i,j}=\sum_{S\in\mathcal{S}^{i,j}}t_S$ //{\em tentative updated trip travel time estimation}\\
\>\>\>\> - {\tt NewRelErr}=$\sum_{T_{i,j}\in\mathcal{T}_{\mathrm{agg}}}\frac{|et'_{i,j}-\bar{tt}_{i,j}|}{\bar{tt}_{i,j}}$ //{\em compute new relative error}\\
\>\>\>\> - if {\tt NewRelErr}$<${\tt RelErr} then do the following //{\em new estimates better than previous ones}\\ 
\>\>\>\>\> - for each $T_{i,j}\in\mathcal{T}_{\mathrm{agg}}$, $et_{i,j}=et'_{i,j}$ //{\em update travel time estimates}\\
\>\>\>\>\> - {\tt RelErr}={\tt NewRelErr}\\
\>\>\>\>\> - {\tt again}={\bf true}; {\bf goto} step 5.2 //{\em perform another iteration}\\
\>\>\>\> - else // {\em new estimates worse than previous ones}\\ 
\>\>\>\>\> {\tt k}=$1+(${\tt k}$-1)\cdot 0.75$ //{\em reduce the street travel time increase/decrease step}\\
\>\>\>\>\> if {\tt k}$<1.0001$ then exit from loop at step 5.3 //{\em if} {\tt k} {\em is too small, exit from inner loop}\\
\>\>\>\>\> else {\bf goto} step 5.3 //{\em otherwise, perform another iteration with smaller} {\tt k}\\
6. \> {\bf Computation of estimated travel time for remaining streets}\\
\>\>- $\mathcal{ES}=\mathcal{S}_{\mathrm{trip}}$; $\mathcal{NS}=\mathcal{S}-\mathcal{S}_{\mathrm{trip}}$\\
\>\>- let $N(S)$ be the set of streets sharing an intersection with street $S$\\
\>\> - for each $S_i\in \mathcal{NS}$ compute $n_{S_i}=|N(S_i)\cap\mathcal{ES}|$\\
\>\> - order the streets in $\mathcal{NS}$ in decreasing order of $n_{S_i}$\\
\>\> - for each $S_i\in \mathcal{NS}$ in the ordered sequence\\
\>\>\> $v_{S_i}=\frac{\sum_{S_j\in N(S_i)\cap\mathcal{ES}}v_{S_j}}{|N(S_i)\cap\mathcal{ES}|}$; $t_{S_i}=\frac{L(S_i)}{v_{S_i}}$; $\mathcal{ES}=\mathcal{ES}\cup\{S_i\}$;$\mathcal{NS}=\mathcal{NS}-\{S_i\}$\\
\>\>- repeat above step until $\mathcal{NS}=\emptyset$\\
6.\> {\bf Travel time estimation}\\
\>\>-for each possible pair of intersections $(I_i,I_j)$, compute optimal route $I_i\to I_j$ \\
\>\>using Dijkstra algorithm with estimated travel time $t_S$ for each street $S$\\
\>\>- let $\mathcal{S}^{i,j}$ be the set of streets in the optimal route for $(I_i,I_j)$\\
\>\>- $ET(i,j)=\sum_{S_h\in\mathcal{S}^{i,j}}\frac{L(S)}{v_S}$\\
\>\>- {\bf return} $ET$
\end{tabbing}
\vspace*{-3mm}
\hrulefill
\end{minipage}}\\[10pt]
Travel time estimation algorithm.
\end{center}

\clearpage

\subsection*{Robustness of day of week (Oracle model)}
To assess whether there is any noticeable difference in terms of trip sharing opportunities between weekend and week days, we have repeated the analysis above for the 104 weekend days. There is no major difference in terms of trip sharing opportunities in weekend versus weekdays. However, the average number of trips per day during weekend days is about $17\%$ lower than that during week days ($\approx350K$ versus $\approx418K$ trips per day). Only minimal differences in total trip travel time savings between week and weekend days are observed, with slightly better savings achieved during weekdays. As shown next, this is due to the strong relation between trip sharing opportunities and the number of performed trips.

\subsection*{Shareable trips versus trips per day (Oracle model)}
To better understand the relationship between trip sharing opportunities and number of trips performed in a day, we have ordered the days for increasing number of performed trips, and plotted the corresponding percentage of shared trips in the day. The resulting plot is reported in Fig.~3C in the main text. Typical days in New York City feature around 400,000 trips with almost near maximum shareability. Days with a small number of trips are rare and happen mostly during special events. For example, the most noticeable drop in trip sharing opportunities occurs on day 240, August 28th 2011, during which hurricane Irene hit New York City. On this (weekend) day, only about $26,500$ trips were performed, and trip sharing opportunities dropped to about $87\%$ when $\Delta=2\,\mathrm{min}$. As such, data points below 300,000 trips per day are too sparse to make statistically reasonable assessments. Hence we have generated additional low density situations by subsampling our dataset, randomly removing various fractions of vehicles from the system in the following way: For each day in the data set, we randomly selected a percentage $c$ of the taxis in the trace, and deleted the corresponding trips from the data set. We varied $c$ from $95\%$ down to $1\%$, generating a number of trips per day as low as 1,962. 

To the set of resulting shareability values we have fit a saturation curve of the form $f(x) = \frac{Kx^n}{1+Kx^n}$, where $K$ and $n$ are two (non-integer) parameters, Fig.~3C in the main text. Curves of this form appear in the well-known Hill equation in biochemistry, describing saturation effects in the binding of ligands to macromolecules and in similar processes \cite{hill1910pea}.
 Fits to both the shared trip maximization and time minimization conditions match very well (for both, $R^2 > 0.99$), we used a standard Levenberg-Marquardt algorithm for obtaining least squares estimates. The best fit parameters read $K=1.1 \times 10^{-4}$, $n=0.92$ for time minimization, and $K=1.5 \times 10^{-6}$, $n=1.39$ for shared trip maximization. Since for time minimization we have $n \approx 1$, the fit works here almost as well (again $R^2 > 0.99$) with the functional form $f(x) = \frac{Kx}{1+Kx}$ that has only the one parameter $K$, known as the Langmuir equation \cite{langmuir1916cfp}, with $K=4.4 \times 10^{-5}$. The Langmuir equation describes the relationship between the concentration of a gas (or compound) adsorbing to a solid surface (or binding site) and the fractional occupancy of the surface. Since an increasing density of taxis -- the ``particles'' -- implies that more trip pairs -- the ``surface'' -- can be covered, the Langmuir equation can thus be seen as an analogy to the saturation effects in shareability if a homogeneous distribution of trips and taxis is assumed. The second parameter $n$ that appears in the Hill equation is used as a measure for cooperative binding in enzyme kinetics: If $n>1$, an enzyme which has already a bound ligand increases its affinity for other ligand molecules. It is unclear if the analogy can be stretched to understand why $n=1.39$ works best for shared trip maximization. In any way, the fast, hyperbolic saturation implies that taxi sharing could be effective even in cities with taxi vehicle densities much lower than New York, and in case of low market penetration of the sharing system a high return on investment.

\subsection*{Increasing the number of shared trips}

We next investigate what happens when we increase the number $k$ of sharable trips from 2 to 3. We remark that the computational complexity of the trip matching task with $k=3$ is orders of magnitude higher than the same task with $k=2$, for the following reasons:
\begin{itemize}
\item[--] The computation of the shareability hyper-network is challenging. In fact, we now have to compare triplets, instead of pairs, of candidate trips. For each triplet, we have 60 possible valid routes connecting the three sources/destinations of the trips, instead of 4 possible routes with $k=2$. For each valid route, we then have to check whether the trips can actually be shared, meaning that the computational time for calculating the shareability network with $k=3$ is at least 15 times higher than that needed to compute the trip sharing graph with $k=2$.
\item[--] We now have to solve a matching problem on hyper-networks, instead of on simple networks. While matching on graphs can be solved in polynomial time, matching on general hyper-networks belongs to the class of NP-hard problems, i.e. problems that are ``difficult to solve". To get around this computational challenge, we use a greedy, polynomial-time heuristic that first builds the maximum matching considering only triplets of trips, then applies standard matching on the remaining trips. This heuristic is known to build a solution which is, in the worst-case, within a constant factor from the optimal solution.
\end{itemize}

To tackle these computational challenges, we computed the number of shared trips and the fraction of saved travel time only in the Online model, and for selected days of the year. Notice that in the Online model trips can be shared only when their starting times are within a temporal window of $\delta$, thus significantly reducing the number of candidate trips for sharing (and, hence, computational time needed to compute the shareability hyper-network) with respect to the Oracle model in which also trips with starting times in time windows larger than $\delta$ can be shared.

We first present the results referring to a day (day 300) in which about $450,000$ trips were performed, which is about the average number of trips per day recorded in our data set. Figures 2D and E in the main text report the percentage of saved taxi trips and of saved travel time as a function of the quality of service parameter $\Delta$, when the time window parameter $\delta$ is set to $1\,\mathrm{min}$. As seen from the figure, increasing the number of sharable trips provides some benefit only when the quality of service parameter $\Delta$ is large enough for such trips taxi sharing opportunities to become available. This value of $\Delta$ is approximately equal to $150\,\mathrm{sec}$. For larger values of $\Delta$, the advantage of triple trip sharing versus double trip sharing becomes perceivable. When $\Delta=300\,\mathrm{sec}$, the number of saved taxi trips is increased from about $50\%$ with $k=2$ to about $60\%$ with $k=3$. While with $k=2$ nearly all trips can be shared, resulting in about halving the number of performed trips, relatively less trips can be combined in a triple trips when $k=3$. In fact, the achieved percentage of saved trips with $k=3$ is $60\%$, which is below (but not too much) the percentage of $66.6\%$ that would result if all trips would be shared in a triple trip.

Similarly to the percentage of saved trips, also when the percentage of total traveled time is considered the difference between double and triple trip sharing becomes perceivable only for values of $\Delta\ge150\,\mathrm{sec}$. Increasing the number of shared trips from 2 to 3 allows a further saving of about $10\%$ in terms of total traveled time, which is achieved when $\Delta=300\,\mathrm{sec}$. 

Figure~S3 reports the percentage of saved taxi trips and of saved total travel time as a function of the time window parameter $\delta$, for two specific values of $\Delta$. While increasing $\delta$ beyond $120\,\mathrm{sec}$ provides little benefits in terms of saved taxi trips when $k=2$, we still can observe some benefit for $\delta>120\,\mathrm{sec}$ with $k=3$. This is due to the fact that with $\delta=120\,\mathrm{sec}$ we already obtain near-ideal performance when $k=2$, corresponding to halving the number of trips. When $k=3$, there is more room for improvement, and a near-ideal performance is approached only when $\delta=180\,\mathrm{sec}$ and $\Delta=5\,\mathrm{min}$. This means that all triple trip sharing opportunities can be exploited only with relatively ``patient" taxi customers. Concerning percentage of saved travel time, we observe that triple trip sharing achieves as far as $45\%$ saving, which is significantly higher than that achieved by double trip sharing. However, similar to the percentage of taxi trips, such high savings can be obtained only with relatively ``patient" taxi customers.

Finally, we compare the potential of triple versus double taxi trip sharing in a relatively less crowded day (day 250), when only $\approx 250,000$ trips were performed. Figure~S3 reports the achieved reduction in number of taxi trips and in total travel time as a function of $\delta$, when $\Delta=5\,\mathrm{min}$. While with $k=2$ near-ideal taxi sharing can be achieved also with low taxi traffic, with $k=3$ a higher number of taxi trip requests is needed to fully exploit the potential of triple trip sharing. The situation is different in terms of saved total travel time, which is consistently benefiting from a higher number of taxi requests for both double and triple trip sharing.

Summarizing, based on the analysis above we can state that triple trip sharing does provide substantial benefits versus double trip sharing, but for this to occur we need a reasonable number of taxi requests, and relatively ``patient" taxi customers, for which waiting for a few minutes at taxi request time and upon arrival at destination is acceptable. With ``impatient" customers, double trip sharing is much more effective than triple trip sharing: it is computationally efficient, and provides nearly the same performance as triple trip sharing. Since the benefits to the community in terms of reduced number of taxis and reduced pollution with triple trip sharing versus double trip sharing are considerable, an interesting question raised by our analysis is whether the New York City municipality can design a fare system that motivates customers to be ``patient".


\clearpage


\begin{table}[ht]
\centering
\small
\begin{tabular}{|c|c|c|c|}
\hline
{\bf Hour} & {\bf Avg. Rel. Error} & $\%$ {\bf trips after filtering} &$|\mathcal{S}_{\mathrm{trip}}|$ \\ 
\hline
0 & 0.1541 & 94.91 & 8582\\
1 & 0.1301 & 96.64 & 8664\\
2 & 0.1433 & 97.76 & 8781\\
3 & 0.1463 & 98.12 & 8673\\
4 & 0.1438 & 98.16 & 8707\\
5 & 0.1448 & 98.17 & 8443\\
6 & 0.1517 & 98.04 & 8485\\
7 & 0.1535 & 98.04 & 8554\\
8 & 0.1560 & 98.01 & 8600\\
9 & 0.1541 & 97.97 & 8596\\
10 & 0.1568 & 97.96 & 8629\\
11 & 0.1644 & 97.85 &8650\\
12 & 0.1639 & 97.97 & 8833\\
13 & 0.1547 & 98.12 & 8820\\
14 & 0.1486 & 98.20 & 8622\\
15 & 0.1553 & 98.19 & 8866\\
16 & 0.1388 & 98.13 & 8687\\
17 & 0.1486 & 98.05 & 8853\\
18 & 0.1457 & 97.86 & 8845\\
19 & 0.1540 & 97.61 & 8718\\
20 & 0.1602 & 97.23 & 8698\\
21 & 0.1649 & 96.85 & 8600\\
22 & 0.1693 & 96.54 & 8641\\
23 & 0.1799 & 95.78 & 8578\\
\hline
{\bf avg} & 0.1534 & 97.59 & 8671.9\\
\hline
\end{tabular}
\caption{Summary of travel time estimation performance.\label{TravelTimeRes}}\vspace{-5mm}
\end{table}


\begin{figure} [ht]
\centering
\includegraphics[width=0.7\textwidth]{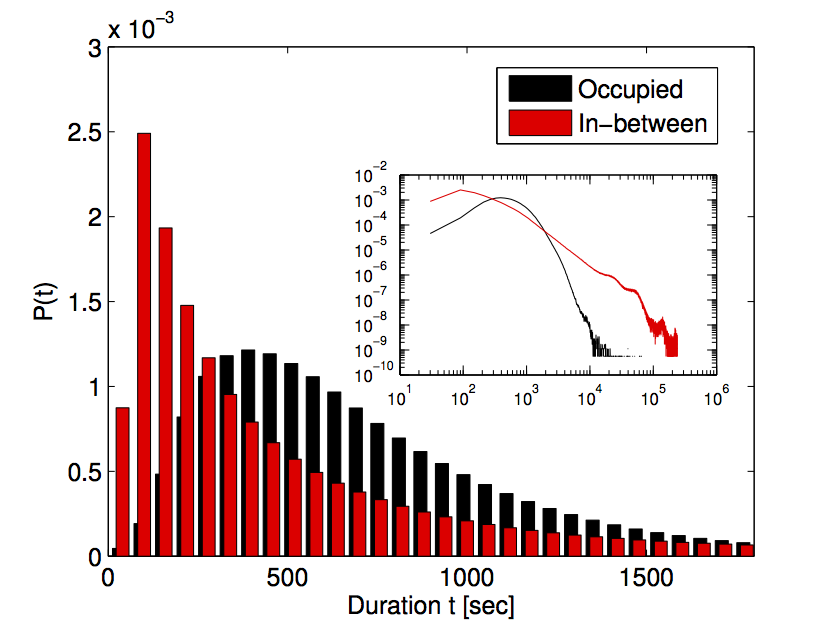}
\caption{Empirical probability distribution of durations of occupied trips, and of the time spans in-between the occupied trips which comprise both empty trips and all activities where taxis are not being used for passenger transport such as shift changes, lunch breaks, vehicle maintenance. Specific data describing the empty trips alone is not available. The inset shows the same two histograms in log-log scale for all measured values.
The in-between durations peak below two minutes, much lower than the durations of occupied trips which peak at around six minutes, showing that taxis tend to find new passengers relatively quickly.} \label{fig:durationdistrib}
\end{figure}

\begin{figure} [ht]
\centering
\includegraphics[width=0.7\textwidth]{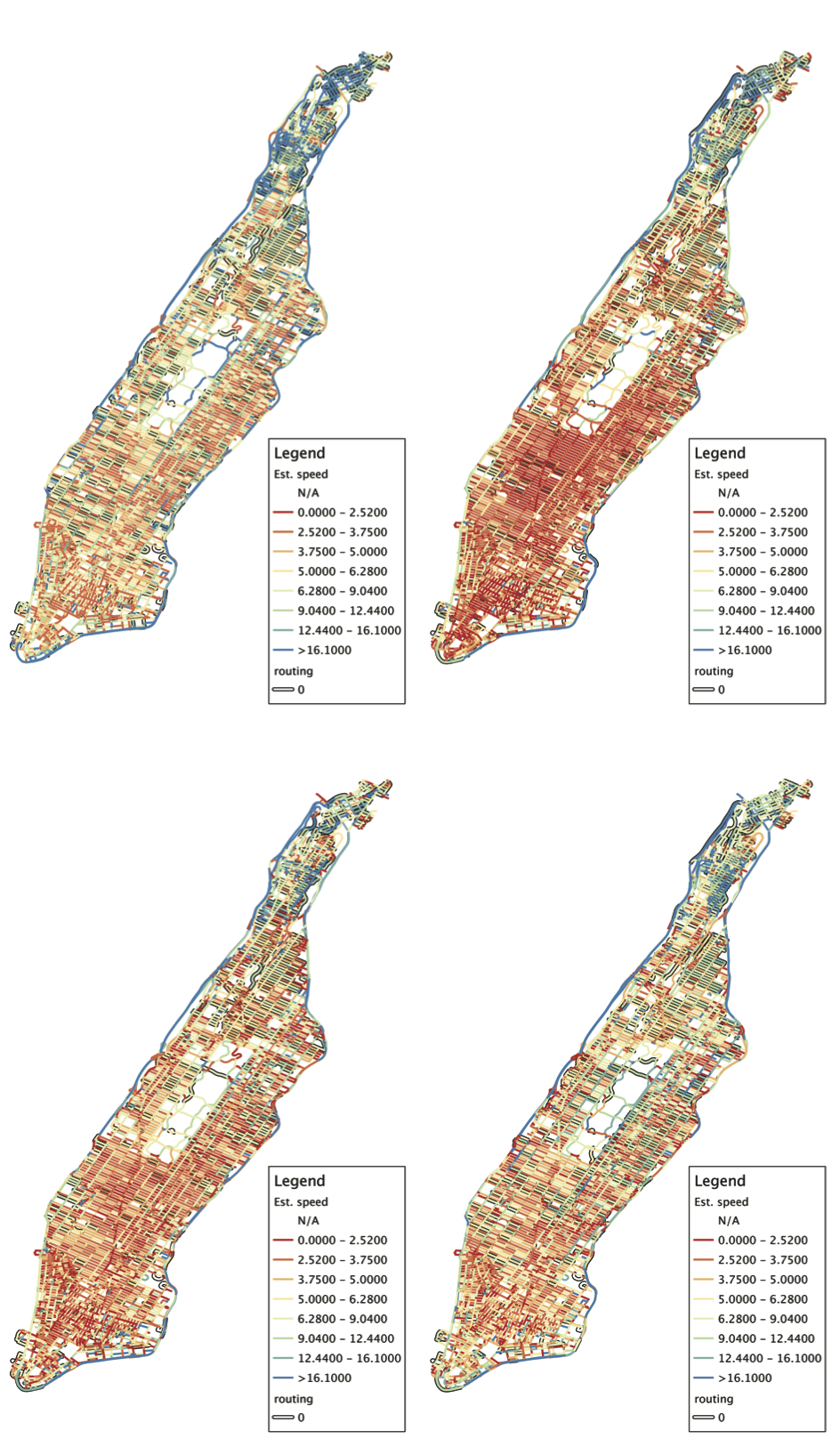}
\caption{Estimated speed map at 0am (top left) and at 8am (top right). Estimated speed map at 4pm (bottom left) and 10pm (bottom right). Travel time for streets in bold (routing) is estimated at step 6 of the algorithm.} \label{speedMap08}
\end{figure}

\begin{figure} [ht]
\centering
\includegraphics[width=0.9\textwidth]{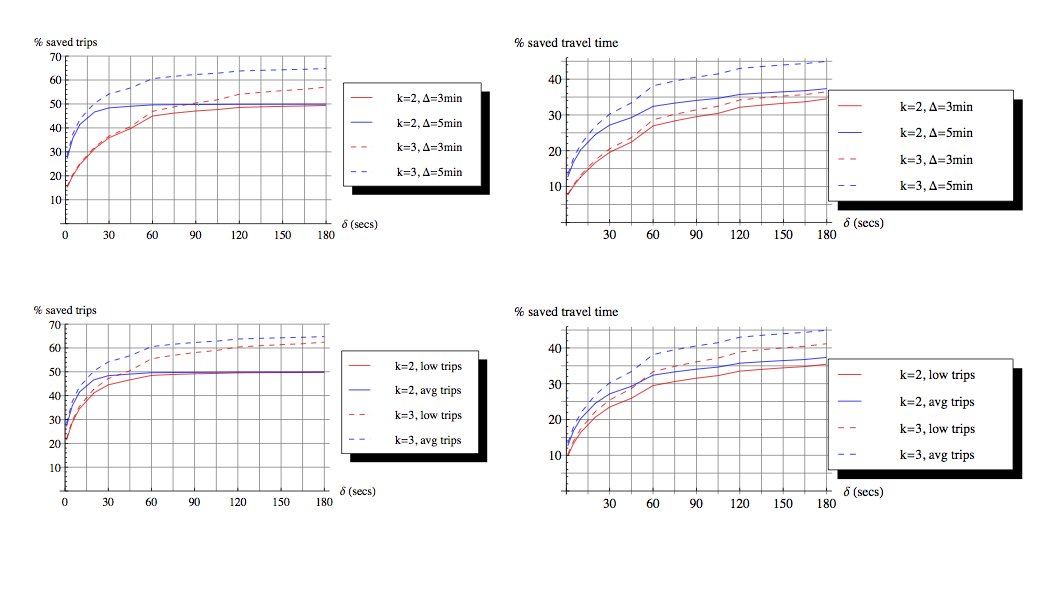}
\caption{Percentage of saved taxi trips (top left) and percentage of saved travel time (top right) in New York City as a function of $\delta$ in the Online model. The quality of service parameter is set to $\Delta=3\,\mathrm{min}$ and $\Delta=5\,\mathrm{min}$. Percentage of saved taxi trips (bottom left) and percentage of saved travel time (bottom right) as a function of $\delta$ in the Online model, in a day with relatively low taxi traffic (``low"), and with average taxi traffic (``avg"). The quality of service parameter is set to $\Delta=5\,\mathrm{min}$. Each plot reports two curves: one referring to the case where at most two trips can be shared ($k=2$), and one referring to the case where at most three trips can be shared ($k=3$).} \label{varKdelta}
\end{figure}